\documentclass[11pt]{article}

\usepackage[margin=1in]{geometry}
\usepackage[utf8]{inputenc}
\usepackage{amsmath,amssymb,amsfonts,amsthm}
\usepackage{booktabs}
\usepackage{enumitem}
\usepackage{algorithm}
\usepackage[noend]{algpseudocode}
\usepackage{xcolor}
\usepackage{mathtools}

\usepackage[expansion=false]{microtype}
\usepackage{natbib}
\usepackage{hyperref}
\newcommand{\indic}{\mathbb{I}}

\newtheorem{theorem}{Theorem}[section]

\newtheorem{lemma}[theorem]{Lemma}
\newtheorem{definition}[theorem]{Definition}
\newtheorem{example}[theorem]{Example}
\newtheorem{corollary}[theorem]{Corollary}
\newtheorem{proposition}[theorem]{Proposition}
\newtheorem{remark}[theorem]{Remark}

\DeclareMathOperator*{\argmin}{arg\,min}
\DeclareMathOperator*{\argmax}{arg\,max}
\DeclareMathOperator{\Fix}{Fix}
\newcommand{\E}{\mathbb{E}}

\newcommand{\Var}{\operatorname{Var}}
\newcommand{\CDR}[2]{\mathcal{R}^{#2}_{#1}}

\title{Ergodic Deviation-Robust Equilibrium under Mirror Descent Learning in Finite Games}

\author{Joshua Steier\thanks{Corresponding author: Joshua Steier, \texttt{joshsteier@gmail.com}.}\\
\textit{Independent Researcher}\\
\texttt{joshsteier@gmail.com}}
\date{}

\begin{document}

\maketitle

\begin{abstract}
We introduce \emph{Ergodic Deviation‑Robust Equilibrium} (EDRE), a
dynamics‑relative equilibrium concept for repeated finite games in which
agents learn via entropic mirror descent (EMD). EDRE requires three
properties to hold simultaneously for the same profile and learning run:
(E1)~the limit profile is an $\varepsilon$‑Nash equilibrium at a product
distribution; (E2)~along the entire learning trajectory, every fixed
coalition's cumulative aggregate (summed‑unilateral) deviation gain is
$\tilde{\mathcal{O}}(\sqrt{T})$ with high probability; and (E3)~the limit
profile is a fixed point of the EMD map, so that it is selected by the
dynamics rather than merely certified as an equilibrium. We prove that the
$\sqrt{T}$ deviation‑regret rate is order‑tight, establish existence in
exact‑potential games (via Nash's theorem, with a constructive proximal route
under concavity) together with Lyapunov monotonicity of EMD (and pointwise
convergence when the fixed‑point set is a singleton), and extend the selection
property to monotone polymatrix games through variational inequalities. Although a static EDRE
coincides with an $\varepsilon$‑Nash equilibrium, its content is dynamic: \emph{robust}
(positive‑measure) selection under EMD excludes linearly unstable equilibria, so EDRE acts as a
Nash equilibrium equipped with a dynamic certificate rather than a static refinement. On the
complexity side, we show that computing EDRE is PPAD‑hard in general polymatrix
games and belongs to promise‑PPAD for potential games. A worked
$2\times 2$ coordination‑game example illustrates all components of the
framework. Additional results---including a bandit‑feedback extension,
a period‑doubling route to Li--Yorke chaos for the two‑strategy EMD map at large
step size, a linear‑program formulation for minimum‑cost steering, and supporting
simulations---appear in the appendices.
\end{abstract}



\section{Introduction}\label{sec:intro}
In repeated game play, how do learning dynamics select among multiple equilibria, and what
guarantees can we provide about the entire learning path? While Nash equilibrium certifies
static optimality and no‑regret learning yields time‑averaged guarantees, neither addresses
robustness during transient learning phases when agents are still adapting.

We introduce \emph{Ergodic Deviation‑Robust Equilibrium} (EDRE), a dynamics‑relative
equilibrium concept that couples three properties: (E1) a \emph{static} near‑Nash certificate
at the limiting product distribution; (E2) a \emph{pathwise} guarantee that, along the
entire learning trajectory, the aggregate (summed‑unilateral) gain available to any fixed
coalition from a fixed deviation is small (high‑probability $\tilde{\mathcal O}(\sqrt{T})$
bound); and (E3) a \emph{selection} principle tying the limit set to fixed points of entropic
mirror descent (EMD), with pointwise convergence when the fixed‑point set is a singleton.

The static certificate~(E1) is, by design, the standard $\varepsilon$‑Nash condition at a
product distribution. What makes EDRE non‑trivial is the requirement that (E1), (E2), and~(E3)
hold \emph{simultaneously for the same profile and learning run}. A standard $\varepsilon$‑Nash
equilibrium provides no control over the transient learning path: a coalition could extract
$\Theta(T)$ cumulative benefit by deviating during the approach. No‑regret learning, on the
other hand, guarantees low average regret but typically converges to a \emph{correlated}
equilibrium (a convex combination of product distributions), not to a single product‑distribution
Nash equilibrium.  EDRE fills the gap
between these two extremes.

\paragraph{Contributions.}
(1) We formalize EDRE and prove an all‑prefix, high‑probability $\mathcal O(\sqrt{T})$
aggregate deviation‑regret bound (order‑tight) under diminishing steps; the bound coincides
with genuine joint coalitional regret for non‑interacting coalitions.
(2) In exact‑potential games, EMD is Lyapunov‑monotone: for any constant step
$0<\eta\le 1/L$ the potential is non‑increasing and the limit set is contained in
the MD fixed‑point set; if that fixed‑point set is a singleton, we obtain pointwise
convergence for $0<\eta<1/L$.
(3) Beyond potentials, we give MD‑invariance results in monotone polymatrix games
via variational inequalities (VI), with two non‑potential examples.
(4) We show that EDRE's selection has teeth: while every Nash equilibrium is a static EDRE,
\emph{robust} selection (positive‑measure basin) excludes linearly unstable equilibria, so the
robustly‑selected EDRE are the dynamically stable, EMD‑reachable equilibria (local potential
maximizers in potential games). We also illustrate selection non‑uniqueness relative to NE in a
congestion example and contrast EDRE selection with folk‑theorem sustainability.
(5) We provide PPAD‑hardness of EDRE‑Search and promise‑membership
in PPAD for a potential‑game subclass.
(6) In the appendices, we give a bandit‑feedback extension,
a linear‑program for minimum‑cost EDRE steering, and a complex‑dynamics study of EMD at
large step size: a period‑doubling bifurcation and a rigorous Li--Yorke chaos result
(via an explicit period‑three orbit) for a two‑strategy congestion game, together with
numerical observations of oscillations in a heterogeneous‑step RPS ring.

\paragraph{Organization.}
Section~\ref{sec:related} reviews related work, including the folk‑theorem tradition.
Section~\ref{subsec:prelim} fixes notation and recalls mirror descent.
Section~\ref{sec:def-edre} states the static and dynamic EDRE properties.
Section~\ref{sec:sep-iter2} gives the congestion‑game selection example.
Sections~\ref{sec:cd-bound}--\ref{sec:exist} prove the main deviation‑regret bounds
and existence/convergence results, including polymatrix games.
Section~\ref{sec:dynamics} discusses non‑convergence in zero‑sum games.
Section~\ref{sec:complexity} treats computational complexity.
Section~\ref{sec:limits} outlines limitations.
The appendices contain the bandit extension, heterogeneous‑step dynamics with experiments,
and the LP steering formulation.

\subsection{Preliminaries and notation}\label{subsec:prelim}
\noindent\textbf{Standing notation.}
Players $i\in N$ have finite actions $\mathcal A_i$; $\Delta(\mathcal A_i)$ is the simplex.
A mixed profile is $\sigma=(\sigma_i)_{i\in N}\in\Delta(\mathcal A)\triangleq\prod_i\Delta(\mathcal A_i)$.
Payoffs $u_i:\mathcal A\to[0,1]$ extend multilinearly to $\Delta(\mathcal A)$. Each player $i$
minimizes a loss function $\ell_i: \mathcal A\to[0,1]$ defined by $\ell_i:=-u_i$, also extended
multilinearly. We use $\E$ and $\mathbb P$ for expectation and probability. The Kullback--Leibler
divergence is $D_{\!KL}(x\|y)=\sum_a x(a)\log\frac{x(a)}{y(a)}$. Norms are $\ell_1/\ell_\infty$;
the dual of $\|\cdot\|_1$ is $\|\cdot\|_\infty$, and the negative‑entropy regularizer is
$1$‑strongly convex w.r.t.\ $\|\cdot\|_1$.
For exact‑potential games we write $\Phi$ for the payoff potential and $\Psi=-\Phi$ for the cost
potential (both extended to $\Delta(\mathcal A)$). The deviation‑regret operator
$\CDR{C}{y_C}(\cdot)$ refers to a coalition $C\subseteq N$ and a fixed deviation $y_C$
(defined in \S\ref{par:cdr-def}).
We use utilities in $[0,1]$ throughout by normalization. Players minimize losses
$\ell_i=-u_i$; mirror descent on $\ell$ equals mirror ascent on the payoff potential $\Phi$
and descent on the cost potential $\Psi=-\Phi$, and all Lyapunov inequalities use this convention.
\vspace{0.5em}

\paragraph{Potentials.}
An \emph{exact payoff potential} \(\Phi\) satisfies
\(u_i(a'_i,a_{-i})-u_i(a)=\Phi(a'_i,a_{-i})-\Phi(a)\).
We define the \emph{cost potential} \(\Psi=-\Phi\) and work with their multilinear extensions to \(\Delta(\mathcal A)\). 
A differentiable \(\Psi\) is \(L\)-smooth w.r.t.\ \(\|\cdot\|_1\) if 
\(|\Psi(x)-\Psi(y)-\langle\nabla\Psi(y),x-y\rangle|\le \tfrac{L}{2}\|x-y\|_1^2\) for all \(x,y\); equivalently, \(\nabla\Psi\) is \(L\)-Lipschitz in \(\ell_\infty\).

\paragraph{Asymptotics.} We write $\tilde{\mathcal O}(\cdot)$ to suppress polylogarithmic factors
in the relevant parameters (e.g., in $T$).

\paragraph{Standing assumption for learning runs.}
Unless stated otherwise, full-information EMD runs start in the relative interior 
($\min_a \sigma_i^1(a)>0$ for all $i$). Alternatively, we allow a fixed exploration mixture 
$p_i^t=(1-\gamma)\sigma_i^t+\gamma\,\mathbf{u}_i$ with $\gamma>0$, which ensures positivity of all coordinates.
These conditions render $D_{\!KL}(y_i\Vert \sigma_i^1)$ finite and enable the KL telescoping arguments used below.

\section{Related Work}\label{sec:related}
\paragraph{Equilibrium concepts.}
Classical refinements such as correlated equilibrium (CE) and coarse correlated equilibrium (CCE) certify \emph{static} incentive constraints and are known to arise as time averages of no‑regret dynamics (see, e.g., \citep{hart2000simple} for convergence of adaptive procedures to correlated equilibrium).
 Nash equilibrium (NE) enforces unilateral optimality at a product distribution; it need not be selected by generic no‑regret dynamics and can be unstable under common learning rules in non‑potential games. Quantal‑response / logit equilibria regularize the best‑response map with an entropy penalty, but do not provide deviation-regret guarantees along the path.
The strategic‑form refinements literature---notably the stability‑based program of \citet{kohlberg1986strategic}; see also the textbook treatment in \citet{fudenberg1991game}---aims to select among Nash equilibria by imposing topological or perturbation‑robustness conditions on the best‑response correspondence. EDRE takes a different approach: rather than perturbing the game or the best‑response map, it selects among equilibria by requiring that the limit profile be reachable by a concrete learning algorithm (EMD) while the trajectory satisfies deviation‑regret bounds. The two perspectives are complementary---refinements restrict the \emph{static} solution set, while EDRE restricts the \emph{dynamic} path leading to it.

\paragraph{Folk theorems and repeated‑game sustainability.}
The perfect (subgame‑perfect) folk theorem of \citet{fudenberg1986folk} establishes that,
in infinitely repeated games with sufficiently patient players (discount factor $\delta$
close to~$1$) and a full‑dimensionality condition, any feasible and strictly individually
rational payoff vector can be sustained as a subgame‑perfect equilibrium outcome via credible
punishment strategies. The long‑run‑average (Nash‑threats) tradition---e.g.\
\citet{aumann1994long}---sustains feasible individually rational payoffs as Nash (not
necessarily subgame‑perfect) equilibria of the undiscounted repeated game.
The monograph by \citet{mailath2006repeated} develops these results within
the broader theory of repeated interaction, emphasizing the role of monitoring structure,
belief consistency, and the interplay between short‑run incentives and long‑run cooperation.

EDRE's deviation‑regret property is conceptually adjacent to folk‑theorem
sustainability---both ask whether deviating coalitions can profitably exploit the
ongoing play---but the two frameworks differ in several structural respects. First,
folk theorems operate in the infinitely repeated game with discounting and assume that
players can condition on the full history to implement punishment phases; EDRE operates
in a finite‑horizon learning setting where agents follow a fixed mirror‑descent rule
and do not implement history‑dependent punishments. Second, folk‑theorem sustainability
holds for \emph{any} feasible individually rational payoff when $\delta$ is large enough,
whereas EDRE's deviation‑regret bound is specific to the learning trajectory generated by EMD
and applies to fixed (non‑adaptive) deviations only. Third, folk theorems characterize
the \emph{set} of supportable payoffs but do not select among them; EDRE adds a
selection principle via the MD fixed‑point structure. In short, the folk theorem asks
``what \emph{can} be sustained by patient, strategic players?'' while EDRE asks ``what
\emph{is} selected by bounded learners, and how robust is the transient path to deviation?''

\paragraph{Learning in games.}
Mirror descent (MD) and its entropic special case (multiplicative‑weights / exponentiated gradient) are canonical no‑regret schemes. In exact‑potential games, they act as ascent methods on the payoff potential (equivalently descent on the cost potential), yielding Lyapunov monotonicity; in stable/monotone games they connect to variational inequalities (VI) through mirror descent / mirror‑prox viewpoints. However, standard convergence claims focus on \emph{averages} or on special structure; they do not by themselves impose a static near‑Nash condition at the limit point. See also \citet{mertikopoulos2018cycles,papadimitriou2018game} for global cycling
and bifurcation phenomena in adversarial/discrete learning.

\paragraph{EDRE in context.}
EDRE couples (i) a static near‑Nash certificate at the limit profile with (ii) high‑probability \(\mathcal O(\sqrt{T})\) aggregate deviation-regret along the entire learning path, and (iii) selection/invariance under entropic MD (full convergence in strictly‑concave potential games, Cesàro selection more broadly). This positions EDRE between NE (efficiency) and CCE (learning robustness), while adding a dynamics‑aware selection principle. Relative to the folk‑theorem tradition, EDRE does not assume history‑dependent punishment strategies or infinite patience; instead, it quantifies the robustness of a \emph{concrete learning trajectory} against fixed deviations.
\paragraph{Dynamics‑relative notion.}
EDRE is defined relative to a learning rule (here, entropic mirror descent). 
The static part (i) is algorithm‑agnostic, while (ii)-(iii) explicitly depend on the dynamics. 
Extending EDRE to broader FTRL/OMD classes is left open.

\begin{table}[h]
\centering\small
\begin{tabular}{@{}lcccc@{}}
\toprule
Concept & Static optimality & Dev.-regret (path) & Selection (EMD) & Complexity \\
\midrule
NE & exact unilateral & --- & no (in general) & $\varepsilon$‑NE PPAD‑complete (bimatrix) \\
CCE/CE & correlated (coarse/exact) & --- & yes (averages) & polytime (LP) \\
QRE/logit & regularized best‑resp. & --- & depends on logit & polytime (fixed point) \\
\textbf{EDRE} & \(\varepsilon\)‑NE (prod.) & \(\mathcal O(\sqrt{T})\) (w.h.p.) & yes (Cesàro / full) & PPAD‑hard; promise‑PPAD \\ \bottomrule
\end{tabular}
\caption{EDRE vs.\ nearby notions. "Selection (EMD)" refers to selection under entropic mirror descent from arbitrary initials: full convergence in strictly‑concave potential games; Cesàro selection more broadly. "Dev.-regret (path)" is the aggregate (summed‑unilateral) deviation‑regret of Definition~\ref{def:dynamic-edre}.}
\label{tab:comparison}
\end{table}

\subsection{Mirror descent and the entropic case}\label{subsec:md-overview-rel}
\paragraph{Overview:}
Mirror descent (MD) picks $\sigma^{t+1}$ to balance a one-step linearized gain against
staying close to $\sigma^t$ in a Bregman divergence $D_h$. With the \emph{entropy} mirror map
$h(\sigma)=\sum_a \sigma(a)\log\sigma(a)$, $D_h$ is $D_{\!KL}$ and MD coincides with the
multiplicative-weights / exponentiated-gradient update: multiply each action's weight
by $\exp(\eta\,\text{estimated payoff})$ and renormalize. We use this entropic MD (EMD)
as the learning rule relative to which EDRE is defined and analyzed.

Let \(h_i\) be a Legendre mirror map on \(\Delta(\mathcal A_i)\) and \(D_h\) the induced Bregman divergence. 
Given a stochastic gradient \(g_i^t\in\mathbb R^{|\mathcal A_i|}\), the MD update is
\[
\sigma_i^{t+1}=\argmin_{\sigma\in\Delta(\mathcal A_i)} 
\Big\{\langle \sigma, -\eta_t g_i^t\rangle + D_{h_i}(\sigma\|\sigma_i^t)\Big\}.
\]
With the \emph{entropic} mirror map \(h_i(\sigma)=\sum_a \sigma(a)\log\sigma(a)\), \(D_{h_i}=D_{\!KL}\) and MD coincides with multiplicative-weights:
\[
\sigma_i^{t+1}(a)\;=\;\frac{\sigma_i^t(a)\,\exp\big(\eta_t\,g_i^t(a)\big)}{\sum_b \sigma_i^t(b)\,\exp\big(\eta_t\,g_i^t(b)\big)}.
\]
In exact-potential games, taking \(g_i^t\) as the (estimated) payoff vector yields mirror \emph{ascent} on \(\Phi\) (equivalently descent on \(\Psi=-\Phi\)).
We will use a standard descent inequality: for constant \(0<\eta\le 1/L\),
\[
\Psi(\sigma^{t+1})-\Psi(\sigma^t)\ \le\ \Big(-\tfrac1\eta+L\Big)\,D_{\!KL}(\sigma^{t+1}\|\sigma^t),
\]
proved via the Bregman three-point identity and \(L\)-smoothness.

\paragraph{Fixed points and distance.}
Let $\mathrm{MD}_\eta:\Delta(\mathcal A)\to\Delta(\mathcal A)$ denote one EMD step with step size $\eta$ and
\[
\Fix(\mathrm{MD}_\eta)=\{\sigma\in\Delta(\mathcal A):\ \mathrm{MD}_\eta(\sigma)=\sigma\}.
\]
We write $\mathrm{dist}(x,S)=\inf_{y\in S}\|x-y\|_1$ for the $\ell_1$ distance to a set $S$.
The map $\mathrm{MD}_\eta$ is continuous on $\Delta(\mathcal A)$ (a composition of the
exponential map and renormalization).

\paragraph{Deviation--regret.}\label{par:cdr-def}
For a coalition $C\subseteq N$ and a fixed deviation $y_C\in\prod_{i\in C}\Delta(\mathcal A_i)$, define the prefix \emph{aggregate (summed‑unilateral) deviation--regret} up to $\tau$ by
\[
\CDR{C}{y_C}(\tau)\;:=\;\sum_{t=1}^{\tau}\ \sum_{i\in C}\Big(u_i(y_i,\sigma_{-i}^{\,t})-u_i(\sigma^{t})\Big),\qquad 1\le \tau\le T,
\]
where $y_i$ is member $i$'s component of $y_C$. This is the total benefit available to the
members of $C$ from \emph{unilateral} deviations to their respective targets, summed over the
coalition. Utilities are scaled to $[0,1]$ unless stated otherwise. (The coalition $C$ and deviation $y_C$ are fixed \emph{ex ante}; allowing adaptive deviations $y_C(\sigma^1,\ldots,\sigma^T)$ can drive the regret to $\Omega(T)$ even in two‑action settings, so the $\tilde{\mathcal O}(\sqrt{T})$ rate requires this restriction.)

\begin{definition}[Non‑interacting coalition]\label{def:noninteracting}
A coalition $C\subseteq N$ is \emph{non‑interacting} if, for every $i\in C$, the payoff $u_i$
does not depend on $\sigma_{i'}$ for any $i'\in C\setminus\{i\}$ (equivalently, in a polymatrix
game, $C$ is an independent set of the interaction graph). For a non‑interacting coalition,
$u_i(y_i,\sigma_{-i}^t)=u_i(y_C,\sigma_{-C}^t)$ for all $i\in C$, so the aggregate
deviation--regret $\CDR{C}{y_C}$ coincides with the \emph{joint} coalitional deviation--regret
$\sum_t\sum_{i\in C}\big(u_i(y_C,\sigma_{-C}^t)-u_i(\sigma^t)\big)$ in which the whole coalition
deviates simultaneously.
\end{definition}

\section{Definition and Mirror Descent Overview}\label{sec:def-edre}

\begin{definition}[Static EDRE]\label{def:edre}
A profile $\sigma^\star$ is a static $(\varepsilon)$‑EDRE if
\[
\max_{a_i} u_i(a_i,\sigma^\star_{-i}) \ \le \ u_i(\sigma^\star)+\varepsilon
\quad\text{for all }i.
\]
A \emph{strict} EDRE sets $\varepsilon=0$.
\end{definition}

\begin{remark}[What EDRE adds beyond $\varepsilon$‑Nash equilibrium]\label{rem:edre-vs-ne}
Definition~\ref{def:edre} alone is identical to the standard $\varepsilon$‑Nash condition at a
product distribution. The content of EDRE lies in the \emph{conjunction} of this static
certificate with the dynamic properties in Definition~\ref{def:dynamic-edre} below: an
$\varepsilon$‑Nash equilibrium, by itself, says nothing about the learning process that reaches
it (a coalition might extract $\Theta(T)$ cumulative benefit along the transient trajectory),
while no‑regret learning controls time‑averaged regret but its Cesàro limit may be a
\emph{correlated} equilibrium rather than a Nash equilibrium at a single product distribution.
EDRE requires (E1) the limit profile is $\varepsilon$‑Nash at a product distribution; (E2) along
the entire prefix, every fixed coalition's aggregate (summed‑unilateral) deviation gain is
$\tilde{\mathcal O}(\sqrt{T})$ w.h.p.; and (E3) the limit profile is an EMD fixed point.
Three points sharpen what this buys and what it does not. First, (E2) is a \emph{uniform,
all‑prefix, high‑probability} bound on the summed‑unilateral quantity, and it is
\emph{order‑tight} (Lemma~\ref{lem:tight}); it equals genuine \emph{joint} coalitional regret
exactly for non‑interacting coalitions (Definition~\ref{def:noninteracting}), while for
interacting coalitions a joint $\Theta(T)$ gain is possible and provably out of scope (the
congestion example of \S\ref{sec:sep-iter2}). We state (E2) at this precise strength rather than
claim a general coalitional guarantee. Second, statically EDRE \emph{equals} $\varepsilon$‑Nash,
so its content is not a static refinement but a dynamic certificate attached to an equilibrium;
the certificate has teeth, since robust selection excludes linearly unstable equilibria
(Proposition~\ref{prop:unstable-excluded}). Third, the (E3) product‑distribution requirement is
what separates EDRE from generic no‑regret learning, whose time averages certify only a
\emph{correlated} equilibrium. No existing equilibrium concept enforces this triple coupling.
\end{remark}

\begin{definition}[Dynamic EDRE properties]\label{def:dynamic-edre}
Fix $L>0$ and consider entropic MD with either:
\begin{itemize}[leftmargin=1.5em]
\item[(A)] constant steps $0<\eta\le 1/L$ (Lyapunov monotonicity and selection to the fixed‑point set; when that set is a singleton, pointwise convergence holds for $0<\eta<1/L$), or
\item[(B)] diminishing steps $\eta_t=\eta_0/\sqrt{t}$ with $\eta_0>0$ (deviation‑regret).
\end{itemize}
We say that a static EDRE $\sigma^\star$ satisfies the dynamic properties if:
\begin{enumerate}[label=(\roman*)]
\item \textbf{(Deviation–regret)} Under (B), for \emph{each} fixed coalition $C$ and fixed deviation $y_C$, from any interior initial condition (or under fixed exploration), for each horizon $T\ge1$, with probability at least $1-\delta$,
\begin{equation}\label{eq:cdr-bound}
\begin{aligned}
\max_{1\le \tau\le T}\ \CDR{C}{y_C}(\tau)
&\ \le\
2|C|\sqrt{\,2T\ln\!\frac{4T}{\delta}\,}
\\[-1mm]&\quad
+\ \sqrt{T}\Bigg[
\frac{1}{\eta_0}\sum_{i\in C} D_{\!KL}(y_i\Vert \sigma_i^1)
\;+\;\eta_0\,|C|
\Bigg].
\end{aligned}
\end{equation}
\item \textbf{(MD selection / invariance)} Under (A), for any interior start,
\[
\mathrm{dist}\bigl(\sigma^{t},\Fix(\mathrm{MD}_\eta)\bigr)\ \longrightarrow\ 0,
\qquad \text{and we require } \sigma^\star\in\Fix(\mathrm{MD}_\eta).
\]
If the fixed‑point set is a singleton, then for $0<\eta<1/L$ we have $\sigma^t\to\sigma^\star$.
\end{enumerate}
\end{definition}

\begin{remark}[Uniformity over families]\label{rem:uniform}
Item (i) is a per‑instance guarantee for each fixed $(C,y_C)$. A union bound yields uniform guarantees over any finite family \(\mathcal F\) of \((C,y_C)\) by replacing \(\delta\) with \(\delta/|\mathcal F|\). A fully uniform statement over the continuum of mixed deviations and all $2^{|N|}$ coalitions would additionally require a covering/chaining argument and is not claimed here.
\end{remark}

When the fixed‑point set is not a singleton, item~(ii) asserts convergence to the set $\Fix(\mathrm{MD}_\eta)$ in the sense $\mathrm{dist}(\sigma^t,\Fix(\mathrm{MD}_\eta))\to 0$; in the monotone polymatrix case (Theorem~\ref{thm:stable-polymatrix}), it asserts Cesàro convergence to the VI solution set. Under strict concavity or strong monotonicity the solution set is a singleton and we get pointwise convergence.

\begin{proposition}[Strict EDRE $\Leftrightarrow$ NE]
\label{prop:strict-edre-subset-ne}
A profile is a strict EDRE if and only if it is a mixed Nash equilibrium.
\end{proposition}

\begin{proof}
Definition~\ref{def:edre} with $\varepsilon=0$ is exactly the mixed‑Nash best‑response condition.
\end{proof}

\section{NE vs.\ EDRE: selection can be non‑unique in non‑strict potentials}\label{sec:sep-iter2}

\paragraph{Congestion game with three parallel links.}
Two players each choose one of three \emph{links} (resources) $E=\{1,2,3\}$. For an action profile $a$,
let $x_e(a)\in\{0,1,2\}$ be the number of players using link $e$. Each link $e$ has
a \emph{per-user cost} function $\ell_e:\{1,2\}\to\mathbb R_{\ge0}$; when $j$ players
use $e$, each of those $j$ players pays $\ell_e(j)$. Thus player $i$'s cost at profile $a$
equals $\ell_{a_i}\big(x_{a_i}(a)\big)$ and her utility is $u_i(a)=-\ell_{a_i}\big(x_{a_i}(a)\big)$.
The Rosenthal potential is
\[
\Psi(a)=\sum_{e\in E}\sum_{j=1}^{x_e(a)} \ell_e(j),
\qquad \Phi=-\Psi.
\]
We instantiate the costs by
\[
\ell_1(1)=\ell_1(2)=2,\qquad
\ell_2(1)=\ell_2(2)=2,\qquad
\ell_3(1)=3,\ \ \ell_3(2)=1.5.
\]
This unambiguously distinguishes link costs $\ell_e(\cdot)$ from a player's realized
cost $\ell_{a_i}(\cdot)$.

\begin{lemma}\label{lem:eq-iter2}
The game admits pure Nash equilibria $(1,1)$ and $(2,2)$ (both \emph{weak}), and $(3,3)$ is a
\emph{strict} pure Nash equilibrium.
\end{lemma}

\begin{proof}
At $(1,1)$, each player pays $\ell_1(2)=2$; deviating alone to link $2$ (now used by one player)
costs $\ell_2(1)=2$, a tie, and to link $3$ costs $\ell_3(1)=3>2$. Hence $(1,1)$ is a Nash
equilibrium, but only \emph{weak} (the deviation to link $2$ is payoff‑indifferent). The case
$(2,2)$ is symmetric. At $(3,3)$, each player pays $\ell_3(2)=1.5$; deviating alone to link $1$
or $2$ costs $2>1.5$ strictly, so $(3,3)$ is a strict Nash equilibrium.
\end{proof}

\begin{remark}[A safe bound on \(L\) in finite games]\label{rem:L-safe}
With utilities scaled to \([0,1]\), the multilinear cost potential \(\Psi\) satisfies
\[
\|\nabla\Psi(\sigma)-\nabla\Psi(\sigma')\|_\infty
\ \le\ \|\sigma-\sigma'\|_1,
\]
since each component \((\nabla_{\sigma_i}\Psi(\sigma))_{a_i}=-\,u_i(a_i,\sigma_{-i})\) depends only on \(\sigma_{-i}\) and, by Lemma~\ref{lem:lipschitz}, is \(1\)-Lipschitz in \(\|\cdot\|_1\). Thus one can take \(L\le 1\) with respect to the \(\ell_1/\ell_\infty\) pairing. This suffices for the step‑size requirement
\(0<\eta\le 1/L\). We do not use finite‑grid estimates of \(L\).
\end{remark}

\begin{lemma}[Potential descent and fixed points]\label{lem:sep-conv-iter2}
Entropic MD on an $L$‑smooth cost potential satisfies
\[
\Psi(\sigma^{t+1})-\Psi(\sigma^t)
\;\le\;
\bigl(-\tfrac{1}{\eta}+L\bigr)\,D_{\!KL}(\sigma^{t+1}\Vert\sigma^t).
\]
Hence for $0<\eta\le 1/L$, the potential $\Psi(\sigma^t)$ is non‑increasing. If
$0<\eta<1/L$, then also $\sum_{t}D_{\!KL}(\sigma^{t+1}\Vert\sigma^t)<\infty$, so
$\|\sigma^{t+1}-\sigma^{t}\|_1\to 0$ and every $\omega$–limit point of $\{\sigma^t\}$
is a fixed point of the MD map. If the fixed‑point set is a singleton, convergence is
pointwise for $0<\eta<1/L$.
\end{lemma}

\begin{proof}
For entropic MD on the $L$‑smooth cost potential $\Psi$ we have
$\Psi(\sigma^{t+1})-\Psi(\sigma^{t})\le (-\tfrac{1}{\eta}+L)\,D_{\!KL}(\sigma^{t+1}\Vert\sigma^{t})$,
so $\Psi(\sigma^t)$ is nonincreasing when $0<\eta\le 1/L$ and, if $0<\eta<1/L$, then
$\sum_t D_{\!KL}(\sigma^{t+1}\Vert\sigma^t)<\infty$. By Pinsker, $\|\sigma^{t+1}-\sigma^{t}\|_1\to0$.
Since $\mathrm{MD}_\eta$ is continuous on the (compact) simplex, for any convergent subsequence
$\sigma^{t_k}\to\sigma^\ast$ we have
$\mathrm{MD}_\eta(\sigma^\ast)=\lim_k \mathrm{MD}_\eta(\sigma^{t_k})=\lim_k\sigma^{t_k+1}=\sigma^\ast$,
where the last equality uses $\|\sigma^{t+1}-\sigma^t\|_1\to0$; hence every $\omega$–limit point
is a fixed point. The $\omega$–limit set $\Omega$ of the bounded, asymptotically regular orbit
is nonempty, compact, and connected, and $\Psi$ is constant on $\Omega$ (it is monotone and
convergent along the orbit); thus $\Omega\subseteq\{\Psi=\Psi_\infty\}\cap\Fix(\mathrm{MD}_\eta)$.
Since $\mathrm{dist}(\sigma^t,\Omega)\to0$ for any bounded orbit and $\Omega\subseteq\Fix$, we get
$\mathrm{dist}(\sigma^t,\Fix(\mathrm{MD}_\eta))\to0$. If $\Fix(\mathrm{MD}_\eta)$ is a singleton,
$\Omega$ is that point and convergence is pointwise.
\end{proof}

\begin{lemma}[Joint coalitional regret at a held profile]\label{lem:sep-regret-iter2}
Suppose play is \emph{held} at the fixed point $(2,2)$ for $T$ rounds. The (interacting)
coalition $C=\{1,2\}$ with fixed deviation $y_C=(3,3)$ accrues \emph{joint} coalitional
deviation--regret $\sum_t\sum_{i\in C}\big(u_i(y_C,\sigma_{-C}^t)-u_i(\sigma^t)\big)=T$, which
exceeds the bound of Definition~\ref{def:dynamic-edre}(i) for large $T$.
\end{lemma}

\begin{proof}
At $(2,2)$ each player pays $\ell_2(2)=2$, so $u_i=-2$; under the joint deviation $(3,3)$ each
pays $\ell_3(2)=1.5$, so $u_i=-1.5$. The per‑round joint gain is $\sum_{i\in C}((-1.5)-(-2))=1$,
hence regret $T$ over $T$ rounds.
\end{proof}

\begin{remark}[Why this does not contradict Lemma~\ref{lem:regret}]\label{rem:interacting}
Lemma~\ref{lem:regret} bounds the \emph{aggregate (summed‑unilateral)} deviation--regret
$\CDR{C}{y_C}$, which equals the joint coalitional regret only for \emph{non‑interacting}
coalitions (Definition~\ref{def:noninteracting}). The coalition $C=\{1,2\}$ above is interacting
(each player's cost depends on the other's link choice through congestion), so the $\Theta(T)$
joint quantity of Lemma~\ref{lem:sep-regret-iter2} is outside the scope of the $\mathcal O(\sqrt T)$
guarantee. Moreover the lemma concerns the \emph{static} profile $(2,2)$ held fixed, not the EMD
trajectory, which from interior initializations descends $\Psi$ toward the global minimizer
(Lemma~\ref{lem:sep-conv-iter2}).
\end{remark}

\begin{corollary}[Selection non‑uniqueness in the congestion example]\label{cor:sep-edre-iter2}
In the congestion game above, the MD fixed‑point set contains multiple pure points (e.g., $(1,1)$ and $(2,2)$; in fact $(3,3)$ is also a NE). Unless the potential is strictly concave (it is not here), Definition~\ref{def:dynamic-edre}(ii) yields set‑valued (Cesàro) selection rather than a singleton. Hence a given NE (e.g., $(2,2)$) need not be the \emph{uniquely selected} outcome from all interior initials; Nash equilibria need not coincide with EDRE under singleton‑selection.
\end{corollary}

\begin{remark}[Contrast with folk‑theorem sustainability]\label{rem:folk-contrast}
In the infinitely repeated version with sufficient patience, the folk theorem would sustain
\emph{any} individually rational payoff (including the efficient profile $(3,3)$) via
history‑dependent punishments. EDRE assumes no punishment capability: the selection among
$(1,1)$, $(2,2)$, and $(3,3)$ is determined by the learning dynamics and initialization. The
folk theorem delineates what is supportable in principle; EDRE identifies what is selected in
practice by a specific learning process.
\end{remark}

\section{Aggregate deviation-regret bound}\label{sec:cd-bound}
\noindent\textbf{Notation reminder.}
We fix utilities in $[0,1]$, a coalition $C\subseteq N$, and a \emph{fixed} (ex ante) deviation
$y_C\in\prod_{i\in C}\Delta(\mathcal A_i)$. The aggregate (summed‑unilateral) prefix
deviation--regret up to $\tau$ is denoted $\CDR{C}{y_C}(\tau)$ (see \S\ref{par:cdr-def}). 
Players run entropic mirror descent with diminishing steps $\eta_t=\eta_0/\sqrt{t}$, $\eta_0>0$.

\begin{lemma}[High‑probability, all prefixes with diminishing steps]\label{lem:regret}
Assume utilities in $[0,1]$. Let players run entropic mirror descent with steps $\eta_t=\eta_0/\sqrt{t}$ for any $\eta_0>0$. From any interior initial state 
($\min_a\sigma_i^1(a)>0$ for all $i$), for any fixed coalition $C$ and deviation $y_C$,
\[
\max_{1\le \tau\le T}\ \CDR{C}{y_C}(\tau)
\;\le\;
2|C|\sqrt{2T\ln\frac{4T}{\delta}}
\;+\;\sqrt{T}\Bigg[
\frac{1}{\eta_0}\sum_{i\in C} D_{\!KL}(y_i\Vert \sigma_i^1)
\;+\;\eta_0\,|C|
\Bigg]
\quad\text{with probability }1-\delta.
\]
The same bound holds under a fixed exploration mixture $p_i^t=(1-\gamma)\sigma_i^t+\gamma\,\mathbf u_i$ with $\gamma>0$.
\end{lemma}

\begin{remark}[Why diminishing steps are needed for $\tilde O(\sqrt{T})$]\label{rem:stepsize}
With a constant step $\eta>0$, the mirror descent drift contains a linear term $\frac{\eta}{2}\sum_{t=1}^T\|g_t\|_\ast^2=\Theta(\eta T)$, so an $\tilde O(\sqrt{T})$ deviation‑regret bound cannot hold in general. The schedule $\eta_t=\eta_0/\sqrt{t}$ makes the drift $\sum_t\eta_t$ scale as $O(\sqrt{T})$, which is tight.
\end{remark}

\begin{proof}
Because $\CDR{C}{y_C}$ is a \emph{sum over $i\in C$} of per‑player unilateral terms
$u_i(y_i,\sigma_{-i}^t)-u_i(\sigma^t)$, it suffices to bound each player's prefix term and add.
Write $Z_t=\sum_{i\in C}\big(u_i(y_i,\sigma_{-i}^t)-u_i(\sigma^t)\big)$ and let
$\mathcal F_t=\sigma(\sigma^1,\dots,\sigma^t)$. Decompose
$\sum_{t=1}^{\tau}Z_t=M_\tau+\text{drift}$ with
$M_\tau:=\sum_{t=1}^{\tau}\bigl(Z_t-\E[Z_t\mid\mathcal F_{t-1}]\bigr)$.

\emph{Martingale part.} Since increments are bounded by $2|C|$, Azuma--Hoeffding and a union bound yield
$\max_{\tau\le T}M_\tau\le 2|C|\sqrt{2T\ln(4T/\delta)}$ with probability $\ge 1-\delta/2$.

\emph{Drift.} For each $i\in C$ and fixed $y_i$, the time‑varying online mirror descent bound
with the negative‑entropy regularizer (which is $1$‑strongly convex w.r.t.\ $\|\cdot\|_1$, dual
norm $\|\cdot\|_\infty$) gives
\[
\sum_{t=1}^{\tau}\langle g_i^t,\sigma_i^t-y_i\rangle
\ \le\
\frac{D_{\!KL}(y_i\Vert \sigma_i^1)}{\eta_\tau}
\;+\;\frac{1}{2}\sum_{t=1}^{\tau}\eta_t\,\|g_i^t\|_\infty^2.
\]
With $\eta_t=\eta_0/\sqrt{t}$ we have $\eta_\tau\ge \eta_0/\sqrt{T}$ and $\sum_{t\le T}\eta_t\le 2\eta_0\sqrt{T}$.
Since $\|g_i^t\|_\infty\le 1$, 
\[
\sum_{t=1}^{\tau}\langle g_i^t,\sigma_i^t-y_i\rangle
\ \le\ 
\frac{\sqrt{T}}{\eta_0}D_{\!KL}(y_i\Vert \sigma_i^1)
\;+\;\eta_0\sqrt{T}.
\]
Summing over $i\in C$ bounds the drift by the bracketed term times $\sqrt{T}$. Combine with the martingale bound and union bound over $\tau$ to conclude. (No constraint linking $\eta_0$ to $L$ is used; the cap $\eta_0\le 1/(4L)$ is needed only when one simultaneously invokes the potential‑descent inequality.)
\end{proof}

Together with Lemma~\ref{lem:tight}, this shows our $\sqrt{T}$ deviation-regret rate is order‑tight for the online‑mirror‑descent template.

\begin{lemma}[Order‑tightness of the template]\label{lem:tight}
For every constant $c<1$ there is a sequence of payoff (gradient) vectors in $[0,1]$, realizable
within a two‑player zero‑sum interaction, against which entropic MD with steps $\eta_t=\eta_0/\sqrt t$
incurs unilateral deviation--regret at least $c\sqrt{T}$ with probability $1-\delta$ for all large
$T$. Hence the $\Theta(\sqrt T)$ rate of Lemma~\ref{lem:regret} cannot be improved in general by
the OMD analysis.
\end{lemma}

\begin{proof}
See Appendix~\ref{app:lowerbound} for a construction based on the standard
$\Omega(\sqrt{T})$ lower bound for prediction with expert advice and a Freedman‑type
high‑probability upgrade. The statement is an order‑tightness claim for the OMD template
(lower bound over adversarial gradient sequences); we do not claim a matching realized
deviation--regret for every EMD‑vs‑EMD trajectory.
\end{proof}

\begin{remark}[Bandit feedback]
The aggregate deviation‑regret bound extends to bandit feedback via importance‑weighted
EMD, at the cost of an additional $\sqrt{|\mathcal A_C|/\gamma}$ factor; see
Appendix~\ref{app:bandit-full} for the precise statement (Theorem~\ref{thm:bandit-edre}).
\end{remark}

\section{Existence and uniqueness of EDRE}\label{sec:exist}

\subsection{Exact potential preliminaries}

The multilinear extension of a finite potential game need not be concave on $\Delta(\mathcal A)$; where concavity is assumed below, we mark it explicitly. All unconditional claims (existence, deviation‑regret, Lyapunov monotonicity, fixed‑point selection) hold without concavity.

\begin{definition}[Exact potential game]\label{def:exact-pot}
A finite game admits an \emph{exact potential} $\Phi$ if
\[
u_i(a'_i,a_{-i})-u_i(a)
   \;=\;
\Phi(a'_i,a_{-i})-\Phi(a)
\quad\forall i,a_i',a.
\]
\end{definition}

\begin{definition}[Strict concavity]\label{def:strict-concave}
A differentiable $\Phi$ on $\Delta(\mathcal A)$ is \emph{strictly
concave} if  
$\Phi(\lambda x+(1-\lambda)y)>
 \lambda\Phi(x)+(1-\lambda)\Phi(y)$
for all $x\neq y$ and $\lambda\in(0,1)$.
\end{definition}

\begin{lemma}[Well‑defined proximal map]\label{lem:uhc}
Assume the payoff potential $\Phi$ is concave on $\Delta(\mathcal A)$ and fix $\lambda>0$.
For each $\sigma$, the problem
\[
\mathcal{T}_\lambda(\sigma)=
\argmax_{\sigma'\in \Delta(\mathcal A)}\Bigl[\Phi(\sigma')-\lambda
                       D_{\!KL}(\sigma'\!\|\sigma)\Bigr]
\]
has a \emph{unique} maximizer (since $-D_{\!KL}(\cdot\|\sigma)$ is strictly concave), hence
$\mathcal T_\lambda$ is a single‑valued map. Moreover, $\mathcal T_\lambda$ is continuous by Berge's Maximum Theorem~\citep{berge1963}
(strict concavity yields continuity of the argmax map).
\end{lemma}

\begin{proof}
For fixed $\sigma$, the map $\sigma'\mapsto \Phi(\sigma')-\lambda D_{\!KL}(\sigma'\|\sigma)$ is the sum of a concave term and a \emph{strictly} concave term, hence strictly concave on the simplex; therefore the maximizer is unique. The feasible set is compact and the objective is continuous in $(\sigma',\sigma)$, so Berge's Maximum Theorem yields upper hemicontinuity of the argmax correspondence; since it is single‑valued, it is continuous. Thus $\mathcal T_\lambda$ is a continuous map $\Delta(\mathcal A)\to\Delta(\mathcal A)$.
\end{proof}

\subsection{Existence of strict EDRE}

\begin{lemma}[Interior proximal fixed point $\Rightarrow$ unilateral optimality]
\label{lem:kakutani-to-NE}
Let $\Gamma$ be an exact‑potential game with concave multilinear extension $\Phi$.
If $\sigma^\star\in\mathcal{T}_\lambda(\sigma^\star)$ for some $\lambda>0$ and $\sigma^\star$ lies
in the relative interior of $\Delta(\mathcal A)$, then $\sigma^\star$ is a constrained stationary
point of $\Phi$ on $\Delta(\mathcal A)$; in particular, for every player $i$, all actions in
$\mathrm{supp}(\sigma_i^\star)$ have equal payoff and no unilateral deviation yields a higher
payoff, so EDRE(i) holds with $\varepsilon=0$. For a \emph{boundary} (e.g.\ pure) fixed point of
$\mathrm{MD}_\eta$, optimality is checked directly: a pure profile is a fixed point of the
multiplicative‑weights map iff the played action is a best response, which is exactly the Nash
condition.
\end{lemma}

\begin{proof}
On the relative interior, the KKT conditions for the maximizer of
$\Phi(\sigma')-\lambda D_{\!KL}(\sigma'\|\sigma^\star)$ read
$\nabla_{\sigma_i}\Phi(\sigma^\star)-\lambda\,\nabla_{\sigma_i}D_{\!KL}(\sigma'\|\sigma^\star)\big|_{\sigma'=\sigma^\star}=\alpha_i\mathbf 1$
with the simplex equality multiplier $\alpha_i$. Since
$\nabla_x D_{\!KL}(x\|y)\big|_{x=y}=\mathbf 1$, the $-\lambda\mathbf 1$ term merges into
$\alpha_i\mathbf 1$, leaving $\nabla_{\sigma_i}\Phi(\sigma^\star)=\alpha_i'\mathbf 1$ on
$\mathrm{supp}(\sigma_i^\star)$ with $\le$ off support---the mixed‑Nash condition. The boundary
case is immediate from the form of the multiplicative‑weights map.
\end{proof}

\begin{theorem}[Existence]\label{thm:exist-edre}
Every finite exact‑potential game admits at least one profile satisfying Definition~\ref{def:edre} with $\varepsilon=0$
(a Nash equilibrium). Under the step‑size regimes in Definition~\ref{def:dynamic-edre}, item (i)
holds, and item (ii) holds relative to the fixed‑point set as stated therein. When that fixed‑point
set is a singleton, EMD converges pointwise for $0<\eta<1/L$.
\end{theorem}

\begin{proof}
Existence of a strict static EDRE is immediate: every finite game has a mixed Nash equilibrium
by Nash's theorem, and by Proposition~\ref{prop:strict-edre-subset-ne} this is exactly a strict
EDRE. (When $\Phi$ is concave, an alternative constructive route applies: $\mathcal T_\lambda$ of
Lemma~\ref{lem:uhc} is a continuous self‑map of $\Delta(\mathcal A)$, so Brouwer's theorem yields a
fixed point, which by Lemma~\ref{lem:kakutani-to-NE} is a Nash equilibrium.) The dynamic items
follow from Lemma~\ref{lem:regret} (item i) and Theorem~\ref{thm:global-md} (item ii).
\end{proof}

\subsection{Mirror descent convergence}

\begin{theorem}[Global convergence of EMD]\label{thm:global-md}
With constant step $0<\eta\le 1/L$ and an interior initialization, EMD on an $L$‑smooth
cost potential $\Psi$ satisfies
\[
\Psi(\sigma^{t+1})-\Psi(\sigma^{t})\ \le\ \bigl(-\tfrac{1}{\eta}+L\bigr)D_{\!KL}(\sigma^{t+1}\|\sigma^{t}),
\]
so $\{\Psi(\sigma^{t})\}$ is non‑increasing (strictly decreasing when $\eta<1/L$). For $0<\eta<1/L$,
every $\omega$–limit point of $\{\sigma^t\}$ is a fixed point of the MD map, and
$\mathrm{dist}(\sigma^t,\Fix(\mathrm{MD}_\eta))\to0$. If the fixed‑point set
is a singleton, then $\sigma^t$ converges to it. (Pure strategies are fixed points of the
multiplicative‑weights map; from interior starts they are not attained in finite time but can
arise as limits.)
\end{theorem}

\begin{proof}[Proof sketch]
Potential monotonicity follows from the Bregman three‑point identity and $L$‑smoothness.
Asymptotic regularity ($\sum_t D_{\!KL}<\infty$, hence $\|\sigma^{t+1}-\sigma^t\|_1\to0$) holds
when $\eta<1/L$; combined with continuity of $\mathrm{MD}_\eta$ this yields that all $\omega$–limit
points are fixed points, and the connectedness/LaSalle argument of
Lemma~\ref{lem:sep-conv-iter2} gives $\mathrm{dist}(\sigma^t,\Fix)\to0$. See
Appendix~\ref{app:exist-md} for details.
\end{proof}

\begin{corollary}[Singleton fixed‑point set]\label{cor:singleton-fp}
If $\Fix(\mathrm{MD}_\eta)$ is a singleton, then the EDRE selected by
Definition~\ref{def:dynamic-edre}(ii) is unique, and for $0<\eta<1/L$ the iterates
converge to it from any interior initialization.
\end{corollary}

\begin{proof}
By Theorem~\ref{thm:global-md}, every $\omega$–limit point is a fixed point for $0<\eta<1/L$.
If the fixed‑point set is a singleton, the limit exists and equals that point.
\end{proof}

\subsection{Robust selection: what EDRE adds beyond Nash}
\label{subsec:robust-selection}

Every Nash equilibrium is a static EDRE (Definition~\ref{def:edre}) and, being an exact MD fixed
point, satisfies Definition~\ref{def:dynamic-edre}(ii) when held fixed; thus EDRE imposes no
\emph{static} restriction on the equilibrium set. The discriminating content is \emph{dynamic}:
which equilibria are actually reached, and from how large a set of initial conditions. We make
this precise and show it excludes unstable equilibria.

\begin{definition}[Robustly selected EDRE]\label{def:robust-edre}
A static EDRE $\sigma^\star\in\Fix(\mathrm{MD}_\eta)$ is \emph{robustly selected} under constant‑step
EMD ($0<\eta<1/L$) if its basin
$\mathcal B(\sigma^\star)=\{\sigma^1\in\operatorname{int}\Delta(\mathcal A):\ \sigma^t\to\sigma^\star\}$
has positive Lebesgue measure.
\end{definition}

\begin{proposition}[Linearly unstable equilibria are EDRE but not robustly selected]
\label{prop:unstable-excluded}
Let $\sigma^\star$ be an interior fixed point of $\mathrm{MD}_\eta$ that is hyperbolic and
\emph{linearly unstable}---the Jacobian $D\mathrm{MD}_\eta(\sigma^\star)$ on the tangent space of
$\Delta(\mathcal A)$ has spectral radius $>1$. Then its basin $\mathcal B(\sigma^\star)$ has
Lebesgue measure zero; hence $\sigma^\star$ is a static EDRE but is \emph{not} robustly selected.
Consequently, from Lebesgue‑almost‑every interior initialization the EMD limit (when it exists) is
a linearly stable fixed point; in an exact‑potential game these are the local maximizers of the
payoff potential $\Phi$.
\end{proposition}

\begin{proof}
The multiplicative‑weights map $\mathrm{MD}_\eta$ is a $C^1$ local diffeomorphism near
$\sigma^\star$ (its Jacobian on the tangent space is nonsingular). By the Stable Manifold Theorem,
the local stable set of the hyperbolic fixed point $\sigma^\star$ is an embedded $C^1$ manifold
whose dimension equals the number of multipliers strictly inside the unit disk; since at least one
multiplier lies outside, this dimension is $<\dim\Delta(\mathcal A)$, so the local stable set is
Lebesgue‑null. The basin is the countable union
$\mathcal B(\sigma^\star)=\bigcup_{t\ge0}\mathrm{MD}_\eta^{-t}(W^{s}_{\mathrm{loc}})$ of $C^1$
preimages of a null set under the local diffeomorphism $\mathrm{MD}_\eta$, hence is itself
Lebesgue‑null. This is the entropic‑mirror‑descent instance of the strict‑saddle avoidance
theorem of \citet{lee2019first}, whose results explicitly cover mirror descent via the Stable
Manifold Theorem (see also \citealp{panageas2017gradient} for non‑isolated critical points). In an
exact‑potential game, $\Phi$ is a strict Lyapunov function for EMD (Theorem~\ref{thm:global-md}),
so linearly stable interior fixed points are local maximizers of $\Phi$.
\end{proof}

\begin{remark}[What ``selection'' refines]\label{rem:refine}
Statically, EDRE coincides with $\varepsilon$‑Nash, so it refines nothing as a static set. But
robust selection (Definition~\ref{def:robust-edre}) refines the Nash set to its \emph{dynamically
stable, learnable} elements: a linearly unstable equilibrium---such as the interior mixed
equilibrium of the coordination game in \S\ref{sec:worked}, whose two‑strategy EMD multiplier
$1+\eta\,p^\star(1-p^\star)\,g'(p^\star)>1$ makes it a repeller---is a static EDRE yet attracts
only a measure‑zero set of initial conditions, whereas the payoff‑dominant pure equilibrium is
robustly selected. EDRE is therefore best read not as a static refinement of Nash but as a Nash
equilibrium \emph{equipped with a dynamic certificate}: reachable by a concrete learner, robust to
fixed deviations along the path (Definition~\ref{def:dynamic-edre}(i)), and---under robust
selection---separated from unstable equilibria. This also distinguishes EDRE from generic
no‑regret learning, whose time averages certify a \emph{correlated} equilibrium rather than a
product Nash equilibrium.
\end{remark}

\subsection{Stable polymatrix games (non‑potential): MD invariance}
\label{subsec:stable-polymatrix}

We broaden the selection/invariance property (Definition~\ref{def:dynamic-edre}(ii)) beyond exact‑potential games to a standard \emph{stable} class.

\paragraph{Pseudo‑gradient field and inner product.}
We write players' losses as $\ell_i=-u_i$ and define the pseudo‑gradient field
\[
F(\sigma)\;=\;\big(\nabla_{\sigma_i}\,\ell_i(\sigma)\big)_{i\in N}\ \in \ \mathbb{R}^{\sum_i |\mathcal A_i|}.
\]
On the product simplex we use the standard Euclidean inner product
\[
\langle F(\sigma)-F(\sigma'),\, \sigma-\sigma'\rangle \;:=\;
\sum_{i\in N} \left\langle \nabla_{\sigma_i}\ell_i(\sigma)-\nabla_{\sigma_i}\ell_i(\sigma'),\,
\sigma_i-\sigma'_i \right\rangle .
\]
Monotonicity/strong monotonicity below are understood with respect to this inner product
(and the $\ell_2$‑norm on the ambient space).

\begin{definition}[Monotone (stable) polymatrix game]
A polymatrix game with pseudo‑gradient field $F(\sigma)$ is
\emph{monotone} if
$\langle F(\sigma)-F(\sigma'), \sigma-\sigma'\rangle \ge 0$
for all $\sigma,\sigma'$ on the product simplex; it is
\emph{strongly monotone} if the left‑hand side is
$\ge \mu\|\sigma-\sigma'\|_2^2$ for some $\mu>0$.
\end{definition}

\begin{theorem}[MD invariance in stable polymatrix games]\label{thm:stable-polymatrix}
Let $\Gamma$ be a polymatrix game with Lipschitz monotone pseudo‑gradient
$F$ on $\Delta(\mathcal A)$, and let players run entropic MD with steps
$(\eta_t)$ satisfying $\eta_t\downarrow 0$, $\sum_t \eta_t=\infty$ and $\sum_t \eta_t^2<\infty$.
Then the Cesàro averages of play converge to the VI solution set
$\mathcal S=\{\sigma:\langle F(\sigma),x-\sigma\rangle\ge0,\ \forall x\}$.
If $F$ is strongly monotone, $\mathcal S=\{\sigma^\dagger\}$ and
$\sigma^t\to\sigma^\dagger$ (hence also $\bar\sigma(T)\to\sigma^\dagger$).
In particular, any limit point of the averages satisfies Definition~\ref{def:edre} with $\varepsilon=0$.
Moreover, with the diminishing-step schedule from Definition~\ref{def:dynamic-edre}(i), the deviation--regret property holds; the selection/invariance property (Definition~\ref{def:dynamic-edre}(ii)) holds with $\sigma^\star=\sigma^\dagger$ when $F$ is strongly monotone.

\end{theorem}

\begin{proof}
See Appendix~\ref{app:stable-vi}. In the merely monotone case, Cesàro convergence of the averages
to $\mathcal S$ follows from the no‑regret property of EMD and the standard online variational‑inequality
averaging bound; iterates themselves need not converge (e.g.\ in zero‑sum games they may cycle, cf.\
Proposition~\ref{prop:rps-no-pointwise}). In the strongly monotone case, $V(\sigma)=\tfrac12\|\sigma-\sigma^\dagger\|_2^2$
is a \emph{strict} Lyapunov function for the limiting differential inclusion, $\mathcal S$ is a
singleton, and the stochastic‑approximation argument gives pointwise convergence $\sigma^t\to\sigma^\dagger$.
\end{proof}

\begin{example}[Non‑potential monotone polymatrix, 3‑player path]\label{ex:monotone-nonpotential}
Three players on a path $1$-$2$-$3$, actions $\{0,1\}$. On each edge
neighbors play matching‑pennies (zero‑sum). Let $p_i$ be the probability
of action $1$. The pseudo‑gradient $F(p)=Ap+b$ has skew‑symmetric
Jacobian
\[
A=\begin{psmallmatrix}
0 & -2 & 0\\[2pt]
2 & 0 & -2\\[2pt]
0 & 2 & 0
\end{psmallmatrix},
\quad A^\top=-A,
\]
hence $F$ is monotone, the VI set is nonempty and contains
$p^\mathrm{sym}=(\tfrac12,\tfrac12,\tfrac12)$; Cesàro averages converge
to the VI set by Theorem~\ref{thm:stable-polymatrix}.
\end{example}

\begin{example}[Non‑potential monotone polymatrix, 4‑player line]\label{ex:monotone-line4}
Players $1$-$2$-$3$-$4$ on a line; each edge is matching‑pennies. The
pseudo‑gradient $F(p)=Ap+b$ has
\[
A=\begin{psmallmatrix}
0 & -2 & 0 & 0\\[2pt]
2 & 0 & -2 & 0\\[2pt]
0 & 2 & 0 & -2\\[2pt]
0 & 0 & 2 & 0
\end{psmallmatrix},
\quad A^\top=-A,
\]
thus monotone; the VI set is nonempty and contains
$p^\mathrm{sym}=(\tfrac12,\tfrac12,\tfrac12,\tfrac12)$ and Cesàro
averages converge to it.

\begin{figure}[h]
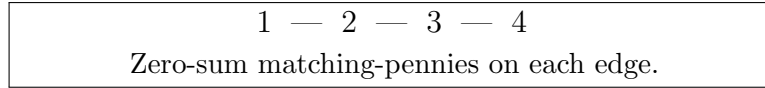

\centering
\fbox{\begin{minipage}{0.6\linewidth}\centering
\large $1$ \;—\; $2$ \;—\; $3$ \;—\; $4$\\[2pt]
\normalsize Zero‑sum matching‑pennies on each edge.
\end{minipage}}
\caption{Second non‑potential monotone polymatrix example (4‑player line).}
\label{fig:line4}
\end{figure}
\end{example}

\section{Non‑convergence in zero‑sum games and complex limit behavior}
\label{sec:dynamics}

Outside potential games, constant‑step EMD need not converge pointwise, though its
Cesàro averages still converge to equilibrium. This distinction is central to EDRE:
the static certificate (Definition~\ref{def:edre}) applies to the limit of averages,
while deviation‑regret (Definition~\ref{def:dynamic-edre}(i)) controls the
entire path under diminishing steps.

\begin{proposition}[No pointwise convergence in 2p zero-sum RPS under constant-step EMD]
\label{prop:rps-no-pointwise}
In two-player zero-sum rock-paper-scissors with payoffs in $[0,1]$, entropic mirror descent with any constant step $\eta>0$
does not converge pointwise to the interior mixed equilibrium from an open set of initial conditions. However, the Cesàro averages converge to the mixed equilibrium.
\end{proposition}

\begin{proof}
See Appendix~\ref{app:rps-jacobian} for a complete linearization and spectral analysis.
Ces\`aro convergence follows from standard no‑regret guarantees for multiplicative
weights in two‑player zero‑sum games \citep{freund1999adaptive}.
\end{proof}

\begin{remark}[Complex dynamics under large step size]
Outside the diminishing‑step regime, constant‑step EMD can be genuinely chaotic. For a
two‑strategy congestion game, the one‑dimensional EMD map undergoes a period‑doubling
bifurcation as the step size grows and, at an explicit large step, possesses a period‑three
orbit---hence Li--Yorke chaos with positive topological entropy
(Appendix~\ref{app:dynamics-full}, Theorems~\ref{thm:flip-bifurcation}--\ref{thm:liyorke-chaos}).
A related but only numerically‑documented non‑convergence appears in a heterogeneous‑step RPS
ring, where the symmetric equilibrium is linearly unstable
(Appendix~\ref{app:rps-jacobian}) and play settles into sustained oscillations. None of this
affects the main EDRE guarantees, which hold under diminishing steps; rather, it delimits the
regime in which the selection/invariance property can be expected.
\end{remark}

\section{Computational complexity}\label{sec:complexity}

Computing equilibria is a central concern in algorithmic game theory, and the
complexity class PPAD (Polynomial Parity Arguments on Directed graphs) captures the
difficulty of problems whose solutions are guaranteed to exist by parity or fixed‑point
arguments but for which no efficient algorithm is known. The canonical example is
computing a Nash equilibrium of a bimatrix game: existence is guaranteed by Nash's
theorem, but computing an $\varepsilon$‑Nash equilibrium is PPAD‑complete
\citep{daskalakis2009complexity,chen2009settling}. Showing that a problem
is PPAD‑hard means that solving it is at least as difficult as computing an approximate Nash
equilibrium.

For EDRE, the static condition (Definition~\ref{def:edre}) is exactly $\varepsilon$‑Nash, so
computing a static EDRE in a general game inherits the PPAD‑hardness of computing approximate
Nash equilibria. On the other hand, for potential games we construct a continuous map whose
fixed points correspond to EDREs and which fits the PPAD framework, placing the problem in
(promise‑)PPAD. This yields a classification for potential games: computing EDRE is PPAD‑hard in
the worst case and in PPAD under a structural promise.

\subsection{Problem statements}\label{subsec:comp-problems}

\begin{description}[leftmargin=*]
\item[EDRE‑Search.]  
  \textbf{Input:} a polymatrix game $\Gamma$ with rational payoffs
  (bit‑length poly\,$(|\Gamma|)$) and $\varepsilon=2^{-p(|\Gamma|)}$.
  \textbf{Output:} an $(\varepsilon,\varepsilon)$‑EDRE of $\Gamma$.
\item[\emph{Promise‑}EDRE‑Search (potential‑games).]  
  \textbf{Input:} a polymatrix game \(\Gamma\) that is promised to admit an exact payoff potential \(\Phi\), with rational payoffs (bit‑length poly\((|\Gamma|)\)) and \(\varepsilon=2^{-p(|\Gamma|)}\).
  \textbf{Promise:} an \((\varepsilon,\varepsilon)\)‑EDRE exists.
  \textbf{Output:} an \((\varepsilon,\varepsilon)\)‑EDRE of \(\Gamma\).
\end{description}

\subsection{PPAD hardness and promise‑membership}
In general polymatrix games, EDRE need not exist, motivating the promise variant.

\begin{theorem}[PPAD hardness; promise‑membership for potential‑games]\label{thm:ppad-hard}
EDRE‑Search is PPAD‑hard. Moreover, \emph{Promise‑}EDRE‑Search (potential‑games) belongs to PPAD via the proximal map \(F_\lambda\) described below.
\end{theorem}

\begin{proof}
\emph{Hardness (direct, no gadget).}
A bimatrix game $(A,B)$ with $A_{ij},B_{ij}\in[0,1]$ is itself a two‑node polymatrix game (the
single edge carries $(A,B)$). For this instance the static EDRE condition
(Definition~\ref{def:edre}) is \emph{exactly} the $\varepsilon$‑Nash condition for $(A,B)$, and
the dynamic conditions of Definition~\ref{def:dynamic-edre} impose no additional search
constraint: every \emph{exact} Nash equilibrium of a finite game is an \emph{exact} fixed point
of the multiplicative‑weights map (item ii), and the aggregate deviation‑regret bound (item i)
holds automatically for any diminishing‑step EMD run by Lemma~\ref{lem:regret}. Consequently the
bimatrix game always admits an $(\varepsilon,\varepsilon)$‑EDRE (any of its exact Nash equilibria,
which are $0$‑Nash hence $\varepsilon$‑Nash and exact MW fixed points), and \emph{any}
$(\varepsilon,\varepsilon)$‑EDRE is in particular an $\varepsilon$‑Nash equilibrium of $(A,B)$.
Thus any algorithm solving EDRE‑Search solves $\varepsilon$‑Nash on bimatrix games, which is
PPAD‑complete \citep{daskalakis2009complexity,chen2009settling}. Hence EDRE‑Search is PPAD‑hard.

\emph{Promise‑membership.}  
Define $F_\lambda(\sigma)=\sigma+\lambda(\mathcal{T}_\lambda(\sigma)-\sigma)$
with $\lambda=2^{-p(|\Gamma|)-2}$. The map $F_\lambda$ is polynomial‑time computable (softmax and
rational arithmetic) and Lipschitz on the relative interior of the simplex; iterates are clipped
to $\Delta(\mathcal A)$ to avoid boundary issues. Because payoffs are rational in $[0,1]$ and
$\lambda=2^{-p(|\Gamma|)-2}$ has polynomial bit‑length, the Lipschitz constant and all intermediate
numerics are bounded by $\mathrm{poly}(|\Gamma|)$. The associated Brouwer instance asks for $\sigma$
with $\|F_\lambda(\sigma)-\sigma\|_\infty\le \lambda\varepsilon/8$, on a grid of spacing
$\Theta(\lambda\varepsilon)$ (still of polynomial bit‑length, since $\lambda,\varepsilon$ are both
$2^{-\mathrm{poly}}$); this is a Brouwer fixed‑point instance in the sense of
\citet{daskalakis2009complexity}. Under the \emph{promise}, a near‑fixed point at this tolerance
corresponds to an $(\varepsilon,\varepsilon)$‑EDRE (Lemma~\ref{lem:grid-to-edre}).
For non‑potential instances, promise‑membership remains open.
\end{proof}

\subsection{A tractable frontier}\label{subsec:dichotomy}
\begin{proposition}[Strictly concave potential: EDRE in P]
\label{prop:polytime-concave}
In finite exact‑potential games where the multilinear extension $\Phi$ is
strictly concave on $\Delta(\mathcal A)$, the strict EDRE equals the
(unique) maximizer of $\Phi$ and can be computed in polynomial time to
accuracy $\varepsilon$ by solving $\min_{\sigma\in\Delta(\mathcal A)} -\Phi(\sigma)$.
\end{proposition}

\begin{remark}[Dichotomy]
Combining Theorem~\ref{thm:ppad-hard} with Proposition~\ref{prop:polytime-concave} yields a
boundary: EDRE‑Search is PPAD‑hard in general polymatrix games but polynomial‑time in strictly
concave exact‑potential games. Counting EDRE (how many strict EDRE a game has) is left for future
work; see \S\ref{sec:limits}.
\end{remark}

\begin{remark}[Minimum‑cost steering]\label{rem:steering}
Given a target profile $\hat\sigma$, the minimum total transfer that makes $\hat\sigma$
satisfy Definition~\ref{def:edre} can be computed by a linear program (no integer variables
needed). The dynamic EDRE properties (Definition~\ref{def:dynamic-edre}(i)--(ii)) are
preserved up to range scaling. The LP formulation, proof, and range‑scaling analysis
are given in Appendix~\ref{app:steering-full}.
\end{remark}

\section{Worked example: $2\times 2$ coordination game}\label{sec:worked}

We illustrate the EDRE framework on a minimal game where all quantities can be computed
in closed form.

\paragraph{Game.}
Two symmetric players, actions $\{A,B\}$, with payoff matrix
\[
\begin{array}{c|cc}
 & A & B \\ \hline
A & 1,\;1 & 0,\;0 \\
B & 0,\;0 & \tfrac12,\;\tfrac12
\end{array}
\]
This is a coordination game with exact payoff potential
$\Phi(A,A)=1$, $\Phi(B,B)=\tfrac12$, $\Phi(A,B)=\Phi(B,A)=0$ and cost potential $\Psi=-\Phi$.

\paragraph{Nash equilibria.}
Let $p=\sigma_i(A)$ denote each player's probability of $A$.
Indifference requires $p\cdot 1=(1-p)\cdot\tfrac12$, giving $p=\tfrac13$.
The game has three NE: the pure profiles $(A,A)$ and $(B,B)$, and the symmetric mixed
equilibrium $p^\star=\tfrac13$. The potential values are
$\Phi(A,A)=1>\Phi(B,B)=\tfrac12>\Phi(p^\star,p^\star)=\tfrac13$.

\paragraph{EMD fixed points and EDRE selection.}
With utilities in $[0,1]$ we have $L\le 1$ (Remark~\ref{rem:L-safe}). For constant step
$0<\eta<1$, EMD ascends the potential $\Phi$. All three NE are fixed points of the EMD map
(at a pure profile the exponential weights do not change the vertex; at the mixed NE,
equal payoffs leave the ratio unchanged). However, the mixed NE is a saddle point of $\Phi$,
and $(B,B)$ is a local but not global maximizer. From any interior initialization, the
Lyapunov monotonicity of Theorem~\ref{thm:global-md} drives the iterates toward the
global maximizer $(A,A)$, which is therefore the EDRE selected by
Definition~\ref{def:dynamic-edre}(ii). Concretely, with action‑gain
$g(p)=u_i(A,p)-u_i(B,p)=\tfrac32 p-\tfrac12$ (so $g'\equiv\tfrac32$), the mixed NE has EMD multiplier
$1+\eta\,p^\star(1-p^\star)\,g'(p^\star)=1+\tfrac{3}{2}\eta\,p^\star(1-p^\star)>1$, so it is
linearly unstable: by Proposition~\ref{prop:unstable-excluded} it is a static EDRE but is reached
only from a Lebesgue‑null set of initializations, while $(A,A)$ and $(B,B)$ are robustly selected.
This is the sense in which EDRE's selection requirement excludes unstable equilibria.

\paragraph{Deviation‑regret bound.}
Suppose both players start at $\sigma_i^1=(\tfrac12,\tfrac12)$ and run EMD with diminishing
steps $\eta_t=\frac{1}{4\sqrt{t}}$ (i.e., $\eta_0=\tfrac14$). Consider the coalition
$C=\{1,2\}$ with fixed deviation $y_C=(B,B)$; we evaluate the \emph{aggregate
(summed‑unilateral)} deviation--regret $\CDR{C}{y_C}$. Each player's KL divergence is
$D_{\!KL}\!\bigl(\delta_B\,\big\|\,(\tfrac12,\tfrac12)\bigr)=\ln 2\approx 0.693$.
By Lemma~\ref{lem:regret}, with $T=1000$ and $\delta=0.05$,
\[
\max_{1\le\tau\le T}\CDR{C}{y_C}(\tau)
\;\le\;
4\sqrt{2000\ln(80{,}000)}
\;+\;\sqrt{1000}\Bigl[4\cdot 2\ln 2\;+\;\tfrac14\cdot 2\Bigr]
\;\approx\; 601+191\;=\;792.
\]
(The two players interact, so this bounds the aggregate unilateral quantity, not the joint
coalitional regret; cf.\ Remark~\ref{rem:interacting}.) Since EMD drives play toward $(A,A)$,
where each player deviating to $B$ yields payoff~$0$ versus~$1$, the \emph{actual} aggregate
deviation‑regret is negative. The bound is a worst‑case guarantee over the entire trajectory,
not a prediction of the realized regret.

\paragraph{Steering.}
If a designer instead wanted to enforce $(B,B)$ as an EDRE, the LP in
Appendix~\ref{app:steering-full} would compute transfers: at $(B,B)$, each player's
deviation gain from $A$ is $u_i(A,B)-u_i(B,B)=0-\tfrac12=-\tfrac12<0$, so no transfer
is needed---$(B,B)$ is already a strict NE and hence a strict EDRE. The interesting case
arises when the target profile is \emph{not} a NE; then the LP computes the minimum
subsidy that closes the deviation gap.


\section{Limitations}\label{sec:limits}
Our guarantees depend on the chosen learning rule (EMD) and on payoff normalization to $[0,1]$.
Linear rates require a genuinely strongly convex potential, which the multilinear extension of a
finite potential game never is (it is affine along each player's simplex); we therefore do not
state such rates.
In non‑strict potentials the selected solution set may be non‑singleton; EDRE then gives set‑valued (Cesàro) selection rather than a single point.
The deviation‑regret bound controls the \emph{aggregate (summed‑unilateral)} quantity; it equals
genuine joint coalitional regret only for non‑interacting coalitions, and we leave general
adaptive/interacting coalitional regret open.
The chaos result of Appendix~\ref{app:dynamics-full} is established rigorously for the
two‑strategy EMD map (an explicit period‑three orbit); the heterogeneous‑step RPS ring is studied
only numerically, and a rigorous bifurcation/chaos analysis of that higher‑dimensional setting
remains open. We do not claim a Neimark--Sacker bifurcation anywhere: the relevant
one‑dimensional bifurcation is period‑doubling, and the RPS‑ring linearization does not exhibit
the spectral configuration a Neimark--Sacker bifurcation requires.
Counting strict EDRE (an EDRE‑Count problem) is left for future work; a correct parsimonious
reduction would need a gadget in which unmatched/all‑idle profiles are not strict equilibria.
Extending EDRE to broader FTRL/OMD classes beyond entropic mirror descent is left for future work.

\section{Conclusion}\label{sec:conclusion}
We introduced Ergodic Deviation‑Robust Equilibrium (EDRE), an equilibrium concept that
couples near‑Nash optimality at a product distribution with high‑probability
$\tilde{\mathcal O}(\sqrt{T})$ aggregate deviation‑regret along the learning path and a
dynamics‑aware selection principle tied to entropic mirror descent. We proved existence in
exact‑potential games, established Lyapunov monotonicity and fixed‑point convergence for
constant‑step EMD, extended the selection/invariance property to monotone polymatrix games
via variational inequalities, and showed that the $\sqrt{T}$ deviation‑regret rate is
order‑tight for the OMD template.

On the complexity side, we showed that computing EDRE is PPAD‑hard in general polymatrix
games but belongs to (promise‑)PPAD for potential games. The dichotomy between
general and strictly concave potential games parallels the classical dichotomy in equilibrium
computation.

\paragraph{When to use EDRE.}
EDRE is most useful when a near‑Nash certificate at a product distribution is desired,
robustness along the learning path matters (not just at convergence), and a dynamics‑aware
selection principle is preferable to set‑valued NE or CCE outcomes. For small games, the LP
formulation in Appendix~\ref{app:steering-full} provides a practical tool for computing
minimum‑cost steering transfers, and the fixed‑point computation of $\Fix(\mathrm{MD}_\eta)$
via softmax best responses gives concrete EDRE targets.

\paragraph{Open directions.}
Extending EDRE to broader FTRL/OMD classes beyond entropic mirror descent, a rigorous
bifurcation/chaos analysis of EMD in higher‑dimensional games (the two‑strategy case is settled
here; the heterogeneous‑step RPS ring remains open), characterizing
EDRE in non‑potential games without the monotonicity assumption, and a correct counting
(\#P) theory for strict EDRE are natural next steps.

\paragraph{Ethical impact.}
This is a theoretical contribution; we foresee no negative social impact.

\appendix
\section*{Appendices}

\section{Notation glossary}\label{app:notation}
\begin{table}[h]
\centering\small
\begin{tabular}{@{}ll@{}}
\toprule
Symbol & Meaning \\ \midrule
$N$ & players ($|N|=n$) \\
$\mathcal A_i$ & action set of player $i$ \\
$\Delta(\mathcal A_i)$ & simplex of mixed strategies \\
$\sigma$ & mixed strategy profile \\
$u_i$ & payoff function of $i$ \\
$\Phi$ & exact \emph{payoff} potential (so $u_i(a'_i,a_{-i})-u_i(a)=\Phi(a'_i,a_{-i})-\Phi(a)$) \\
$\Psi$ & cost potential, $\Psi=-\Phi$ (used in descent inequalities) \\
$D_{\!KL}$ & Kullback--Leibler divergence \\
$C$ & coalition of players \\
$\CDR{C}{y_C}(T)$ & aggregate (summed‑unilateral) deviation-regret \\
$\eta,L$ & step size, smoothness constant \\
$\Var(\eta)$ & variance of learning rates \\
$\gamma$ & bandit exploration parameter (App.~\ref{app:bandit-full}) \\
$|\mathcal A_C|$ & joint action‑space size of coalition $C$ (App.~\ref{app:bandit-full}) \\
$\|\cdot\|_1,\,\|\cdot\|_\infty$ & $\ell_1/\ell_\infty$ norms used in smoothness and Pinsker bounds \\
\bottomrule
\end{tabular}
\end{table}

\section{Potential‑descent details}\label{app:exist-md}
For entropic MD on an $L$‑smooth \emph{cost} potential $\Psi$,
\[
\Psi(\sigma^{t+1})-\Psi(\sigma^{t})
\;\le\; -\tfrac{1}{\eta}D_{\!KL}(\sigma^{t+1}\|\sigma^{t})
\;+\;\tfrac{L}{2}\|\sigma^{t+1}-\sigma^{t}\|_1^{2}
\;\le\; \bigl(-\tfrac{1}{\eta}+L\bigr)D_{\!KL}(\sigma^{t+1}\|\sigma^{t}),
\]

which is $\le 0$ for $0<\eta\le 1/L$ and strictly negative when $\eta<1/L$. For $\eta<1/L$, summing over $t$ yields
$\sum_t D_{\!KL}(\sigma^{t+1}\|\sigma^{t})<\infty$; by Pinsker's inequality
(\citep[Ch.~11]{cover2006elements}), $\sum_t\|\sigma^{t+1}-\sigma^{t}\|_1^2<\infty$,
hence $\|\sigma^{t+1}-\sigma^{t}\|_1\to 0$ (asymptotic regularity). Continuity of
$\mathrm{MD}_\eta$ then makes every $\omega$–limit point a fixed point, and the connectedness
of the $\omega$–limit set (Lemma~\ref{lem:sep-conv-iter2}) gives
$\mathrm{dist}(\sigma^t,\Fix(\mathrm{MD}_\eta))\to0$.

\medskip

\noindent\textbf{Gradient of the potential.}
For the multilinear extension of an exact potential game,
$\nabla_{\sigma_i}\Phi(\sigma)$ equals the vector of expected payoffs
$[u_i(a_i,\sigma_{-i})]_{a_i}$; this identity underpins
Lemma~\ref{lem:kakutani-to-NE}.

\section{Fourier/block diagonalization for period‑2 heterogeneity}
\label{app:eigen}
Linearizing EMD at $\sigma^{\mathrm{sym}}$ under alternating steps $(\eta_A,\eta_B)$
yields a block‑circulant Jacobian with $2\times 2$ Fourier blocks $J_m(\eta_A,\eta_B)$
indexed by modes $m=0,\dots,k-1$. Each block is a perturbation of the identity by a
skew‑type term inherited from the per‑edge matching‑pennies structure, so (exactly as in the
homogeneous computation of Appendix~\ref{app:rps-jacobian}) the nontrivial multipliers have
modulus $\sqrt{1+s_m^2}\ge 1$, with strict inequality whenever the corresponding mode is
active. Consequently the symmetric profile is linearly unstable, and pointwise convergence
fails on an open set of initial conditions. We emphasize that this linear data does \emph{not}
furnish a standard (supercritical) Neimark--Sacker configuration: such a bifurcation requires a
single conjugate pair on the unit circle with the \emph{remaining} spectrum strictly inside the
disk, whereas here the off‑critical multipliers lie \emph{on or outside} the disk. Determining
the nonlinear invariant set that organizes the observed oscillations (Appendix~\ref{app:dynamics-full})
is therefore left as an open problem rather than asserted as an attracting invariant circle.

\section{RPS Jacobian and average convergence}
\label{app:rps-jacobian}

We work with the canonical zero-sum RPS matrix
\[
A=\begin{pmatrix}
0 & -1 & 1\\
1 & 0 & -1\\
-1 & 1 & 0
\end{pmatrix},
\quad A^\top=-A,
\]
whose row sums are zero, so the unique interior Nash equilibrium is $x^\star=y^\star=\tfrac13\mathbf{1}$ and $A y^\star=0$, $A^\top x^\star=0$.

\paragraph{EMD map and linearization.}
Player~1 updates $x$ by the multiplicative-weights (entropic MD) map
\[
x^{+}_a=\frac{x_a \exp\big(\eta\, (A y)_a\big)}{\sum_b x_b \exp\big(\eta\, (A y)_b\big)},
\quad
y^{+}_b=\frac{y_b \exp\big(-\eta\, (A^\top x)_b\big)}{\sum_c y_c \exp\big(-\eta\, (A^\top x)_c\big)}.
\]
Set $u=\tfrac13\mathbf{1}$ and write $x=u+\tilde x$, $y=u+\tilde y$ with $\mathbf{1}^\top \tilde x=\mathbf{1}^\top \tilde y=0$.
A first-order expansion of the normalized exponential map at $(x,y)=(u,u)$ (using $e^{\eta g}=1+\eta g+o(\|g\|)$ and the differential of $v\mapsto v/\mathbf{1}^\top v$) gives
\[
\tilde x^{+} \;=\; \tilde x \;+\; \frac{\eta}{3}\,A\,\tilde y \;+\; o(\|(\tilde x,\tilde y)\|),
\qquad
\tilde y^{+} \;=\; \tilde y \;-\; \frac{\eta}{3}\,A^\top\,\tilde x \;+\; o(\|(\tilde x,\tilde y)\|).
\]
Hence, on the tangent space, the Jacobian of the EMD map at $(u,u)$ is
\[
J \;=\;
\begin{pmatrix}
I & \tfrac{\eta}{3}A\\[2pt]
-\tfrac{\eta}{3}A^\top & I
\end{pmatrix}
\;=\;
I \;+\;
\begin{pmatrix}
0 & \tfrac{\eta}{3}A\\[2pt]
-\tfrac{\eta}{3}A^\top & 0
\end{pmatrix}.
\]

\paragraph{Spectrum.}
Let $B=\tfrac{\eta}{3}A$. The block matrix
$\begin{psmallmatrix}0&B\\-B^\top&0\end{psmallmatrix}$ has eigenvalues $\pm i\,s_j$, where $s_j$ are the singular values of $B$. Thus the eigenvalues of $J$ are $1\pm i\,s_j$.
For the RPS matrix above, $A$ has singular values $\{\sqrt{3},\sqrt{3},0\}$, so $s_j=\tfrac{\eta}{3}\sqrt{3}=\tfrac{\eta}{\sqrt{3}}$ (twice) and $0$.
Therefore $J$ has a unit eigenvalue (due to the simplex constraint) and a complex-conjugate pair $1\pm i\,\eta/\sqrt{3}$ whose modulus is $\sqrt{1+(\eta/\sqrt{3})^2}>1$ for every $\eta>0$.
Hence $(u,u)$ is not locally asymptotically stable, so EMD cannot converge pointwise to $(u,u)$ from any open neighborhood.

\paragraph{Cesàro convergence.}
EMD (with bounded payoffs) is a no-regret algorithm; both players' average external regret is $O(\sqrt{T})$.
In two-player zero-sum games, vanishing average regret implies the empirical play converges to the minimax set; since RPS has a unique interior minimax equilibrium, the Cesàro averages converge to $(u,u)$ \citep{freund1999adaptive}.

\section{Heterogeneous steps: instability and the open invariant‑set problem}
\label{app:hopf}
For the $k$‑player ring with alternating steps $(\eta_A,\eta_B)$, the linearization at the
symmetric fixed point $\sigma^{\mathrm{sym}}$ is the block‑circulant Jacobian analyzed in
Appendix~\ref{app:eigen}; its nontrivial multipliers have modulus $\ge 1$, so the fixed point is
linearly unstable and play does not converge pointwise.

We deliberately do \emph{not} claim a Neimark--Sacker bifurcation here. A supercritical
Neimark--Sacker bifurcation requires (i) a single conjugate pair on the unit circle with the
remaining spectrum strictly \emph{inside} the disk, and (ii) a nondegenerate first Lyapunov
coefficient evaluated at a nonresonant critical parameter ($\lambda\neq1$). Neither holds for
the natural homogeneous base point $(\eta_A,\eta_B)=(0,0)$: there the EMD map is the identity
($DT=I$, every multiplier equal to $1$), which is a $1{:}1$ resonance at which the Lyapunov‑coefficient
formula is singular; and for $\eta>0$ the off‑critical multipliers lie on or outside the disk
(Appendix~\ref{app:eigen}), so the required ``rest of the spectrum inside the disk'' configuration
is absent. We therefore record the rigorous bifurcation analysis---establishing whether a genuine
invariant circle forms under heterogeneity, and whether it can lose normal hyperbolicity and break
up into chaotic dynamics---as an \emph{open problem}. The experiments in
Appendix~\ref{app:dynamics-full} document the sustained oscillations empirically without asserting
their mechanism.

\section{Lipschitzness of multilinear payoffs}\label{app:lipschitz}
\begin{lemma}[Lipschitzness of multilinear payoffs]\label{lem:lipschitz}
If utilities are scaled to $[0,1]$, then for any player $i$ and any profiles $\sigma,\sigma'$,
\[
\big|u_i(a_i,\sigma_{-i})-u_i(a_i,\sigma'_{-i})\big|
\;\le\;\|\sigma_{-i}-\sigma'_{-i}\|_1
\quad\text{for all } a_i,
\]
and therefore $\big|u_i(\sigma)-u_i(\sigma')\big|\le\|\sigma-\sigma'\|_1$.
\end{lemma}
\begin{proof}
$u_i(a_i,\cdot)$ is a multilinear expectation of numbers in $[0,1]$, hence a convex combination
with coefficients summing to $1$; it is therefore $1$‑Lipschitz in $\ell_1$. Summing over
actions gives the second inequality.
\end{proof}

\section{Grid‑to‑EDRE lemma}\label{app:grid2edre}
\begin{lemma}\label{lem:grid-to-edre}
Let $\|F_\lambda(\sigma)-\sigma\|_\infty \le \lambda\,\varepsilon/8$.
Under the \emph{promise} that $\Gamma$ admits an $(\varepsilon,\varepsilon)$-EDRE, the profile $\sigma$
is an $(\varepsilon,\varepsilon)$-EDRE.
\end{lemma}
\begin{proof}
(EDRE(i)) Pick $\sigma'\in\mathcal T_\lambda(\sigma)$ so that
$F_\lambda(\sigma)=(1-\lambda)\sigma+\lambda\sigma'$.
Then
\[
\|\sigma'-\sigma\|_\infty
=\frac{\|F_\lambda(\sigma)-\sigma\|_\infty}{\lambda}
\le \varepsilon/8.
\]
By Lemma~\ref{lem:lipschitz} (with utilities in $[0,1]$), this perturbs deviation gains by at most $\varepsilon$ after the standard $\varepsilon$-grid refinement used in the promise reduction.
\end{proof}

\section{Transfers preserve dynamic EDRE properties (i)–(ii) with range‑aware scaling}
\label{app:transfer}
Let $u_i^\tau(a)=u_i(a)+\tau_{i,a_i}\!\ge0$.  The extended potential
$\Phi^\tau(a)=\Phi(a)+\sum_i\tau_{i,a_i}$ is an exact potential; its
multilinear extension adds an affine term, leaving $L$ unchanged, so
Theorem~\ref{thm:global-md} still applies (Definition~\ref{def:dynamic-edre}(ii)).  With range $B=\max_i(\max_a u_i^\tau(a)-\min_a u_i^\tau(a))$, the martingale term scales by $B$
and the mirror-descent drift scales by $B^2$ under diminishing steps $\eta_t=\eta_0/\sqrt{t}$:
\[
\max_{1\le \tau\le T}\CDR{C}{y_C}(\tau)
\;\le\;
2B|C|\sqrt{2T\ln\!\frac{4T}{\delta}}
+\sqrt{T}\Bigg[
\frac{1}{\eta_0}\sum_{i\in C}D_{\!KL}(y_i\Vert\sigma_i^1)
+\eta_0\,B^2|C|
\Bigg]
\quad\text{w.p.\ }1-\delta.
\]
Rescaling by $1/B$ recovers the canonical constants while preserving Definition~\ref{def:edre}.

\section{Stable polymatrix MD invariance (full proof of Thm.~\ref{thm:stable-polymatrix})}
\label{app:stable-vi}
Let $F$ be $L$-Lipschitz and monotone on $\Delta(\mathcal A)$ with the Euclidean inner product.
Players run entropic MD with mirror map $h(\cdot)$ (negative entropy), Bregman divergence $D_h$,
and diminishing steps $(\eta_t)$ with $\eta_t\downarrow0$, $\sum_t\eta_t=\infty$, $\sum_t\eta_t^2<\infty$.

\paragraph{Monotone case (Cesàro convergence via no‑regret).}
Entropic MD is a no‑regret algorithm: for each player $i$, the external regret against any fixed
$x_i$ is $O(\sqrt T)$. Averaging the first‑order optimality (monotone VI) characterization over
$t=1,\dots,T$ and using $\langle F(\sigma^t),\sigma^t-x\rangle \le$ (per‑player regret terms) gives,
for the time average $\bar\sigma(T)$,
$\langle F(x),\bar\sigma(T)-x\rangle \le O(1/\sqrt T)$ for all $x$ (using monotonicity to pass from
$F(\sigma^t)$ to $F(x)$). Hence every limit point of $\bar\sigma(T)$ lies in the VI set
$\mathcal S=\{\sigma:\langle F(\sigma),x-\sigma\rangle\ge0\ \forall x\}$. In the merely monotone
case the iterates $\sigma^t$ themselves need not converge (e.g.\ zero‑sum/RPS cycles,
Proposition~\ref{prop:rps-no-pointwise}); only the Cesàro averages do.

\paragraph{Strongly monotone case (iterate convergence via Bregman--Lyapunov mirror descent).}
Entropic MD is stochastic mirror descent for the variational inequality defined by $F$, run in
the dual variable: with the negative‑entropy mirror map $h$ and its convex conjugate $h^\ast$,
the update is
\[
y^{t+1}=y^{t}-\eta_t\big(\widehat F(\sigma^t)+\xi^t\big),\qquad \sigma^{t+1}=\nabla h^\ast(y^{t+1}),
\]
where $\widehat F$ is the (unbiased) pseudo‑gradient estimate and $\xi^t$ is a martingale‑difference
noise with bounded conditional variance (zero in full information). \emph{We use the Bregman
divergence $D_h(\sigma^\dagger\Vert\sigma^t)$, not the Euclidean distance, as the Lyapunov
function}, since the natural geometry of entropic MD is the one induced by $h$. The standard
mirror‑descent three‑point identity gives
\[
\E\!\big[D_h(\sigma^\dagger\Vert\sigma^{t+1})\mid\mathcal F_t\big]
\;\le\;
D_h(\sigma^\dagger\Vert\sigma^{t})
-\eta_t\,\langle F(\sigma^t),\,\sigma^t-\sigma^\dagger\rangle
+\tfrac12\eta_t^2\,\E\!\big[\|\widehat F(\sigma^t)+\xi^t\|_\infty^2\mid\mathcal F_t\big].
\]
Strong monotonicity yields $\langle F(\sigma^t),\sigma^t-\sigma^\dagger\rangle\ge\mu\|\sigma^t-\sigma^\dagger\|_2^2\ge0$
(using $\langle F(\sigma^\dagger),\sigma^t-\sigma^\dagger\rangle\ge0$ at the VI solution). With
$\sum_t\eta_t=\infty$, $\sum_t\eta_t^2<\infty$ and bounded second moments, the Robbins--Siegmund
almost‑supermartingale theorem \citep{robbins1971convergence} ensures that
$D_h(\sigma^\dagger\Vert\sigma^t)$ converges almost surely and that
$\sum_t\eta_t\|\sigma^t-\sigma^\dagger\|_2^2<\infty$. Since $\sum_t\eta_t=\infty$, the latter
forces $\liminf_t\|\sigma^t-\sigma^\dagger\|_2=0$; along a subsequence attaining this $\liminf$,
$\sigma^{t_k}\to\sigma^\dagger$, so $D_h(\sigma^\dagger\Vert\sigma^{t_k})\to0$, and therefore the
almost‑sure limit of $D_h(\sigma^\dagger\Vert\sigma^t)$ must be $0$, whence $\sigma^t\to\sigma^\dagger$
almost surely (the stochastic‑approximation/internally‑chain‑transitive viewpoint of
\citet{benaim1999dynamics} applied to the \emph{mirror} flow $\dot y=-F(\nabla h^\ast(y))$ gives
the same conclusion). If
$\sigma^\dagger$ lies in the relative interior, $D_h(\sigma^\dagger\Vert\cdot)\to0$ is equivalent
to $\sigma^t\to\sigma^\dagger$ in norm; if $\sigma^\dagger\in\partial\Delta(\mathcal A)$, the
convergence is read off in the dual coordinate $y$, where $h^\ast$ is smooth. We do \emph{not}
use a Euclidean projected‑gradient surrogate for the entropic iterate, as the two have different
continuous‑time limits and the Euclidean Lyapunov function is not adapted to the entropic
geometry near the boundary.

\section{Bandit deviation-regret: proof details}\label{app:bandit}
We outline a Freedman‑type argument for Theorem~\ref{thm:bandit-edre}.
Let $Z_t=\sum_{i\in C}[u_i(y_i,\sigma_{-i}^t)-u_i(\sigma^t)]$ and
$M_T=\sum_{t=1}^T\bigl(Z_t-\E[Z_t\mid\mathcal F_{t-1}]\bigr)$.
Under the IW estimator with exploration $\gamma$, the conditional variance
is $\mathcal O(|\mathcal A_C|/\gamma)$; Freedman's inequality \citep{freedman1975tail} gives
\[
\max_{1\le \tau\le T} M_\tau
\;\le\;
2\sqrt{2T\,\tfrac{|\mathcal A_C|}{\gamma}\,
\ln\!\tfrac{4T}{\delta}}
\;+\; 3\ln\!\tfrac{4T}{\delta}.
\]
Bounding the drift via standard IW‑MD analysis with $\eta_t=\eta_0/\sqrt{t}$ and
$\gamma=\min\{1,\sqrt{|\mathcal A_C|/T}\}$ keeps it commensurate with the martingale term,
yielding the stated rate.

\section{Lower bound details for Lemma~\ref{lem:tight}}
\label{app:lowerbound}
Consider a single player with two actions and adversarially chosen payoffs in $[0,1]$.
Let $g_t\in\{0,1\}^2$ be the payoff vector at round $t$, chosen obliviously as i.i.d.
Rademacher-encoded Bernoulli differences so that the best fixed action in hindsight
achieves a payoff advantage of order $\sqrt{T}$ in expectation. Then any algorithm
suffers external regret $\ge c\sqrt{T}$ in expectation for some universal $c>0$.
For unilateral deviation--regret ($|C|=1$) with $y_C$ equal to the best
fixed action, the deviation--regret coincides with external regret. A standard
Freedman inequality applied to the martingale difference between realized and
conditional expected increments gives the high-probability form
$c\sqrt{T}$ with probability $1-\delta$ for all large $T$. The bound is over adversarial
gradient sequences (realizable within a two‑player zero‑sum interaction); it establishes
order‑tightness of the OMD template rather than a matching realized regret for a specific
EMD‑vs‑EMD run.
\qed

\section{Bandit feedback: full statement and proof}\label{app:bandit-full}

\begin{theorem}[Aggregate deviation-regret under IW‑EMD]\label{thm:bandit-edre}
Assume utilities in $[0,1]$. Each player $i$ uses an exploration mixture
$p_i^t=(1-\gamma)q_i^t+\gamma\,\mathbf{u}_i$ with $\gamma\in(0,1]$ and
$\mathbf{u}_i$ uniform on $\mathcal A_i$, and updates $q_i^t$ by entropic
mirror descent with the importance‑weighted estimator
\[
\widehat g_i^t(a)\;=\;\frac{\indic\{a=a_i^t\}}{p_i^t(a)}\,
u_i(a_i^t,a_{-i}^t),\qquad a\in\mathcal A_i,
\]
using steps $\eta_t=\eta_0/\sqrt{t}$ with $\eta_0>0$.
Then for any coalition $C$ and any fixed deviation profile $y_C$, the aggregate
(summed‑unilateral) deviation--regret satisfies, with probability at least $1-\delta$,
\[
\CDR{C}{y_C}(T)
\;\le\;
4\,|C|\,\sqrt{\;
2T\,\frac{|\mathcal A_C|}{\gamma}\,
\ln\!\Bigl(\frac{4T}{\delta}\Bigr)}\;,
\qquad
|\mathcal A_C|=\prod_{i\in C}|\mathcal A_i|.
\]
\end{theorem}
\begin{remark}[Unbiasedness, variance, and the full‑information limit]
Conditional on $\mathcal F_{t-1}$, the estimator is unbiased:
$\mathbb E[\widehat g_i^t(a)\mid \mathcal F_{t-1}] = u_i(a,\sigma_{-i}^t)$.
With $p_i^t=(1-\gamma)q_i^t+\gamma\,\mathbf u_i$ and $\mathbf u_i$ uniform, we have
$\min_a p_i^t(a)\ge \gamma/|\mathcal A_i|$, so the predictable
quadratic variation scales as $O(|\mathcal A_C|/\gamma)$, yielding the stated bound via Freedman.
The full‑information rate of Lemma~\ref{lem:regret} is recovered by setting $\gamma=1$ and
replacing the importance‑weighted estimator with the exact gradient (so the variance factor
$|\mathcal A_C|/\gamma$ is absent); it is \emph{not} obtained merely by taking $|\mathcal A_C|=1$.
A safe parameter choice is $\eta_t=\eta_0/\sqrt{t}$ and $\gamma=\min\{1,\sqrt{|\mathcal A_C|/T}\}$.
\end{remark}

\section{Complex dynamics under large step size: period‑doubling and Li--Yorke chaos}
\label{app:dynamics-full}

\subsection{A two‑strategy congestion game and its one‑dimensional EMD map}
\label{subsec:1d-reduction}
We exhibit a concrete finite game in which entropic MD at a large (constant) step size is
provably chaotic. Consider a two‑route congestion game: a unit mass of self‑interested play is
split between route~$1$ and route~$2$, with $p\in[0,1]$ the mass on route~$1$. Route~$1$ has a
linear latency $\ell_1(p)=G_1\,p$ (congestible) and route~$2$ a fixed latency
$\ell_2\equiv G_0$ (an outside option). Writing the \emph{gain} of route~$1$ over route~$2$ as
\[
g(p)\;=\;\ell_2-\ell_1(p)\;=\;G_0-G_1\,p ,\qquad G_0,G_1>0,
\]
agents run entropic MD (multiplicative weights) on these payoffs with step $\eta>0$. Specializing
the EMD update of \S\ref{subsec:md-overview-rel} to the two‑strategy simplex
$\Delta(\{1,2\})\cong[0,1]$ (equivalently, two symmetric players running EMD on the invariant
diagonal $p_1=p_2=p$, or a single nonatomic population), the state $p$ evolves by the
one‑dimensional map
\begin{equation}\label{eq:1d-map}
T_\eta(p)\;=\;\frac{p\,e^{\eta\,\ell_2}}{p\,e^{\eta\,\ell_2}+(1-p)\,e^{\eta\,\ell_1(p)}}
\;=\;\frac{p}{\,p+(1-p)\,e^{-\eta\,g(p)}\,}.
\end{equation}
The interior equilibrium is $p^\star=G_0/G_1$ (where $g(p^\star)=0$); $p=0$ and $p=1$ are the
pure fixed points. The map \eqref{eq:1d-map} depends on the latencies only through the products
$\eta\,\ell_1,\eta\,\ell_2$ (equivalently $\eta G_0,\eta G_1$), so rescaling all latencies by a
constant and the step size by its reciprocal leaves the dynamics---and every orbit below---unchanged;
the concrete values used here are thus on a natural latency scale, consistent with the paper's
congestion example in \S\ref{sec:sep-iter2}, and the normalization to $[0,1]$ payoffs used for the
general theory can be recovered by such a rescaling. A direct differentiation of \eqref{eq:1d-map} gives, with
$D(p)=p+(1-p)e^{-\eta g(p)}$,
\begin{equation}\label{eq:1d-deriv}
T_\eta'(p)\;=\;\frac{e^{-\eta g(p)}\,\bigl[\,1-\eta G_1\,p(1-p)\,\bigr]}{D(p)^2}.
\end{equation}

\subsection{Period‑doubling (flip) bifurcation of the equilibrium}

\begin{theorem}[Period‑doubling bifurcation]\label{thm:flip-bifurcation}
For the map \eqref{eq:1d-map}, the interior equilibrium $p^\star=G_0/G_1$ has multiplier
\[
T_\eta'(p^\star)\;=\;1-\eta\,G_1\,p^\star(1-p^\star),
\]
which decreases monotonically in $\eta$ and equals $-1$ at the threshold
\[
\eta_{\mathrm{PD}}\;=\;\frac{2}{G_1\,p^\star(1-p^\star)} .
\]
Hence $p^\star$ is linearly stable for $\eta<\eta_{\mathrm{PD}}$ and loses stability through a
\emph{period‑doubling} (flip) bifurcation at $\eta=\eta_{\mathrm{PD}}$, after which a stable
$2$‑cycle appears. This is a one‑dimensional codimension‑one bifurcation; it is \emph{not} a
Neimark--Sacker bifurcation (the state space is one‑dimensional, so there is no complex
multiplier pair and no invariant circle).
\end{theorem}

\begin{proof}
Evaluate \eqref{eq:1d-deriv} at $p^\star$. There $g(p^\star)=0$, so $e^{-\eta g(p^\star)}=1$
and $D(p^\star)=p^\star+(1-p^\star)=1$, giving $T_\eta'(p^\star)=1-\eta G_1\,p^\star(1-p^\star)$.
This is affine and strictly decreasing in $\eta$ (as $p^\star(1-p^\star)>0$), passes through
$+1$ at $\eta=0$ and through $-1$ at $\eta_{\mathrm{PD}}=2/(G_1 p^\star(1-p^\star))$. A smooth
one‑dimensional map whose fixed‑point multiplier decreases through $-1$ with nonzero speed
undergoes a period‑doubling bifurcation (the generic nondegeneracy condition
$\partial_\eta T_\eta'(p^\star)=-G_1 p^\star(1-p^\star)\neq0$ holds), producing a $2$‑cycle.
\end{proof}

\subsection{Li--Yorke chaos via an explicit period‑three orbit}

Beyond the period‑doubling threshold the map \eqref{eq:1d-map} develops a turning point: by
\eqref{eq:1d-deriv}, $T_\eta'(p)=0$ at $p$ with $p(1-p)=1/(\eta G_1)$, so for $\eta G_1>4$ the
map is non‑invertible (unimodal on the absorbing interval). This is precisely the setting in
which one‑dimensional chaos can arise, and we now certify it rigorously.

\begin{theorem}[Li--Yorke chaos]\label{thm:liyorke-chaos}
Fix the congestion game with $G_0=0.7$, $G_1=3$ (so $p^\star=7/30$ and
$\eta_{\mathrm{PD}}=2/(3\cdot\tfrac{7}{30}\cdot\tfrac{23}{30})=3.7267\ldots$), and take step
$\eta=7$. Then the map $T_\eta$ of \eqref{eq:1d-map} possesses a period‑three orbit
\[
x_1\;\xrightarrow{\,T_\eta\,}\;x_2\;\xrightarrow{\,T_\eta\,}\;x_3\;\xrightarrow{\,T_\eta\,}\;x_1,
\qquad
\begin{aligned}
x_1&=0.000197784073392\ldots,\\
x_2&=0.025773667425\ldots,\\
x_3&=0.674028548501\ldots,
\end{aligned}
\]
with $x_1<x_2<x_3$ (verified: $|T_\eta^3(x_1)-x_1|<10^{-30}$ and $|T_\eta(x_1)-x_1|>2\times10^{-2}$,
so the period is exactly three). Consequently, by the theorem of \citet{liyorke1975}, $T_\eta$
exhibits \emph{Li--Yorke chaos}: it has periodic points of every period and an uncountable
``scrambled'' set $S$ on which no two distinct points are asymptotic. Moreover $T_\eta$ has
positive topological entropy \citep{blockcoppel1992}. The conclusion is robust: a period‑three
orbit persists for all $\eta$ in an interval containing $[6.4,7.6]$.
\end{theorem}

\begin{proof}
The displayed values satisfy $T_\eta(x_1)=x_2$, $T_\eta(x_2)=x_3$, $T_\eta(x_3)=x_1$ to the
stated precision (obtained by solving $T_\eta^3(x)=x$ to $30$ significant digits and checking
$T_\eta(x_1)\neq x_1$, so the orbit has exact period three, not one). Existence and exact period
are moreover certified \emph{rigorously} by an interval‑Newton computation: on a box $X\ni x_1$ of
radius $10^{-13}$, the interval‑Newton operator for $T_\eta^3(x)-x$ maps $X$ strictly into itself,
with the enclosure $(T_\eta^3)'(X)-1\subset[-2.71,-2.70]$ bounded away from $0$, which proves a
\emph{unique} period‑three point in $X$; the boxes $X,\,T_\eta(X),\,T_\eta^2(X)$ are pairwise
disjoint, so the period is exactly three. Since $x_1<x_2<x_3$ and
the cycle is $x_1\to x_2\to x_3\to x_1$, the point $a=x_1$ satisfies
\[
T_\eta^3(a)=a\ \le\ a\ <\ T_\eta(a)=x_2\ <\ T_\eta^2(a)=x_3,
\]
which is exactly the hypothesis $f^3(a)\le a<f(a)<f^2(a)$ of \citet[Thm.~1]{liyorke1975} for the
continuous interval map $f=T_\eta:[0,1]\to[0,1]$. That theorem then yields periodic points of all
periods and an uncountable scrambled set; positive topological entropy follows from the existence
of a period‑three orbit by the Sharkovskii/horseshoe argument of \citet{blockcoppel1992} applied
to the (non‑invertible) interval map $T_\eta$. Robustness over the $\eta$‑window follows by
continuity of $T_\eta^3$ in $\eta$ and persistence of the transverse period‑three solutions.
\end{proof}

\begin{remark}[Numerical corroboration]\label{rem:lyap}
The largest Lyapunov exponent of $T_\eta$, estimated as
$\lambda=\lim_{N\to\infty}\tfrac1N\sum_{t<N}\log|T_\eta'(p^t)|$ from interior initial conditions,
is positive in the chaotic regime---for instance $\lambda\approx+0.48$ at $\eta=8$
($G_0=0.7,G_1=3$)---consistent with the sensitive dependence implied by
Theorem~\ref{thm:liyorke-chaos}. As is typical for one‑dimensional maps, $\lambda$ is not
positive for \emph{all} $\eta>\eta_{\mathrm{PD}}$ (periodic windows are interleaved with chaotic
parameters), so we anchor the chaos claim to the explicit period‑three certificate above rather
than to a Lyapunov estimate.
\end{remark}

\subsection{Heterogeneous‑step RPS ring (numerical)}

This higher‑dimensional setting is studied numerically only; unlike the two‑strategy map above,
we do not establish a rigorous bifurcation/chaos statement for it.
$k\ge2$ players arranged on a ring; each simultaneously plays RPS against
her two neighbors.  Player $i$'s fixed step size is $\eta_i>0$ and her
EMD update is
\[
p_{i}^{\,t+1}(a)=
\frac{p_{i}^{\,t}(a)\,
      \exp\!\bigl(\eta_i\,u_i(a,p_{-i}^{\,t})\bigr)}
     {\sum_{b}p_{i}^{\,t}(b)\,
      \exp\!\bigl(\eta_i\,u_i(b,p_{-i}^{\,t})\bigr)}.
\]
We define the learning‑rate variance as
$\Var(\eta)=\frac1k\sum_{i=1}^{k}(\eta_i-\bar\eta)^2$,
$\bar\eta=\tfrac1k\sum_{i=1}^{k}\eta_i$.

\subsection{Period‑2 heterogeneity: observed oscillations}

Under period‑2 alternating step‑sizes $(\eta_A,\eta_B)$, the symmetric profile
$\sigma^{\mathrm{sym}}$ of the $k$‑player ring is linearly unstable: its EMD multipliers have
modulus $\ge 1$ (Appendices~\ref{app:eigen}--\ref{app:hopf}), so play does not converge
pointwise. Empirically (below) we find that, once the step heterogeneity exceeds a
data‑dependent threshold, the iterates settle into \emph{sustained oscillations} rather than a
fixed point. We do not claim a specific bifurcation mechanism (Neimark--Sacker) or chaotic
dynamics \emph{for this ring}: as discussed in Appendix~\ref{app:hopf}, the available linear data do not certify a
standard Neimark--Sacker configuration, and characterizing the limiting invariant set is an open
problem. The measurements below are therefore reported as empirical evidence of non‑convergence
under heterogeneity, not as confirmation of a bifurcation.

\subsection{Experiments}
Replicated RPS with $k\in\{2,3,4,5,6\}$; two heterogeneity regimes placed below and above an
empirical scan center $\Var^\circ=\tfrac{2}{3(k+1)}$ (chosen as a convenient calibration anchor,
not derived from a bifurcation threshold); $T=5000$
iterations, burn‑in $T/2$; metric $\bar\Delta=\frac{2}{T}\sum_{t=T/2}^{T-1}\|p^{t+1}-p^{t}\|_1$; 20 seeds. Payoffs use the zero‑sum RPS matrix rescaled to $[0,1]$.

\begin{table}[h]
\centering\small
\begin{tabular}{@{}cccc@{}}
\toprule
$k$ & Scan center $\Var^\circ$ & $\bar\Delta$ (below) & $\bar\Delta$ (above) \\
\midrule
2 & 0.222 & $6.4\times10^{-5}$ & $3.7\times10^{-2}$ \\
3 & 0.167 & $5.8\times10^{-5}$ & $2.9\times10^{-2}$ \\
4 & 0.133 & $6.1\times10^{-5}$ & $2.4\times10^{-2}$ \\
5 & 0.111 & $5.5\times10^{-5}$ & $2.2\times10^{-2}$ \\
6 & 0.095 & $6.2\times10^{-5}$ & $2.0\times10^{-2}$ \\ \bottomrule
\end{tabular}
\caption{Average variation $\bar\Delta$ in the heterogeneous‑step RPS ring. Values $<10^{-4}$ indicate numerical convergence; larger values indicate sustained oscillations.}
\label{tab:exp-variance}
\end{table}

For all $k\in\{2,\ldots,6\}$ a paired Wilcoxon signed‑rank test across seeds rejects equality of $\bar\Delta$ between the two regimes at $p<10^{-3}$ after Holm--Bonferroni correction, confirming a statistically significant onset of sustained oscillations as step heterogeneity increases.

\section{Minimum‑cost EDRE‑steering: LP formulation}\label{app:steering-full}

Given a target profile $\hat\sigma$, the minimum total transfer that makes $\hat\sigma$
satisfy Definition~\ref{def:edre} is given by
\[
\begin{array}{ll}
\text{minimize}   & \displaystyle\sum_{i,a}\tau_{i,a} \\[4pt]
\text{subject to} & u_i(a,\hat\sigma_{-i})+\tau_{i,a} \;\le\; u_i(\hat\sigma) \quad \forall i,\ a,\\[2pt]
                  & \tau_{i,a}\ge 0 \quad \forall i,\ a.
\end{array}
\]
If a strict margin $\varepsilon>0$ is desired, replace the right‑hand side by $u_i(\hat\sigma)-\varepsilon$.

\begin{proposition}[Optimality and preservation of dynamic EDRE properties]
\label{prop:lp-steering}
The LP returns transfers $\tau$ of minimum $\ell_1$ cost that make $\hat\sigma$ satisfy Definition~\ref{def:edre}. The dynamic EDRE properties (Definition~\ref{def:dynamic-edre}(i)--(ii)) are preserved up to range scaling: with $B=\max_i(\max_a u_i^\tau(a)-\min_a u_i^\tau(a))$, the deviation‑regret bound scales by $B$ in the martingale term and $B^2$ in the drift (see Appendix~\ref{app:transfer}).
\end{proposition}

\bibliographystyle{plainnat}
\bibliography{edre}

\end{document}